\documentclass[a4paper,10pt]{article}
\usepackage[utf8x]{inputenc}
\usepackage{amsmath, amssymb, amsthm}
\usepackage{fullpage, times}
\usepackage{url}
\usepackage{ifpdf}

\ifpdf
    \usepackage[pdftex]{hyperref}
\else
    \usepackage[hypertex]{hyperref}
\fi

\title{Decentralized Dynamics for Finite Opinion Games\thanks{Work partially supported by EPSRC grant EP/G069239/1 ``Efficient Decentralised Approaches in Algorithmic Game Theory''
and by PRIN 2008 research project COGENT (COmputational and GamE-theoretic aspects of uncoordinated NeTworks), funded by the Italian Ministry of University and Research.
A preliminary version of this paper appeared as \cite{SAGT12}.
}}
\author{Diodato Ferraioli\thanks{LAMSADE, Universit\`e Pars Dauphine, France. E-mail: {\tt diodato.ferraioli@dauphine.fr}.} \and Paul W. Goldberg\thanks{Department of Computer Science, University of Liverpool, UK. E-mail: {\tt P.W.Goldberg@liverpool.ac.uk}.} \and Carmine Ventre\thanks{School of Computing, Teesside University, UK. E-mail: {\tt C.Ventre@tees.ac.uk}.}}

\newtheorem{theorem}{Theorem}[section]
\newtheorem{claim}[theorem]{Claim}
\newtheorem{lemma}[theorem]{Lemma}
\newtheorem{cor}[theorem]{Corollary}
\newtheorem{obs}[theorem]{Observation}
\newtheorem{definition}[theorem]{Definition}

\newcommand{\x}{{\mathbf x}}
\newcommand{\y}{{\mathbf y}}
\newcommand{\z}{{\mathbf z}}
\newcommand{\w}{{\mathbf w}}
\newcommand{\blf}{{\mathbf b}}
\newcommand{\0}{{\mathbf 0}}
\newcommand{\1}{{\mathbf 1}}
\newcommand{\G}{\mathcal{G}}
\newcommand{\Ham}{\mathcal{H}}
\newcommand{\OO}{{\mathcal O}}

\newcommand{\SC}{{\sf SC}}
\newcommand{\cw}{{\sf CW}}
\newcommand{\diam}{\text{\rm diam}}

\newcommand{\tm}{{t_\text{\rm mix}}}

\newcommand{\tr}{{t_\text{\rm rel}}}
\newcommand{\pimin}{\pi_{\text{\rm min}}}
\newcommand{\tv}[1]{\left\|#1\right\|_{\rm TV}}
\newcommand{\tc}{{\tau_\text{\rm couple}}}

\newcommand{\poly}[1]{{\sf poly}\left(#1\right)}

\newcommand{\Prob}[2]{\mathbf{P}_{#1} \left( #2 \right)}
\newcommand{\Expec}[2]{\mathbf{E}_{#1} \left[ #2 \right]}

\newcounter{todo}

\def\pbre{pure best-response equivalent\ }

\renewcommand{\epsilon}{\varepsilon}

\begin{document}
\date{}

\setcounter{page}{0}
\maketitle
\thispagestyle{empty}

\begin{abstract}
Game theory studies situations in which strategic players can modify the state of a given system, due to the absence of a central authority. Solution concepts, such as Nash equilibrium, are defined to predict the outcome of such situations.
In multi-player settings, it has been pointed out that to be realistic, a solution concept should be obtainable via processes that are decentralized and reasonably simple. Accordingly we look at the computation of solution concepts by means of decentralized dynamics. These are algorithms in which players move in turns to improve their own utility and the hope is that the system reaches an ``equilibrium'' quickly.

We study these dynamics for the class of opinion games, recently introduced by Bindel et al.~\cite{bkoFOCS11}. These are games, important in economics and sociology, that model the formation of an opinion in a social network. We study best-response dynamics and show upper and lower bounds on the convergence to Nash equilibria. 
We also study a noisy version of best-response dynamics, called logit dynamics, and prove a host of results about its convergence rate as the noise in the system varies. To get these results, we use a variety of techniques developed to bound the mixing time of Markov chains, including coupling, spectral characterizations and bottleneck ratio.
\end{abstract}


\newpage
\setcounter{page}{1}
\section{Introduction}
Social networks are widespread in physical and digital worlds. The following scenario therefore becomes of interest. Consider a group of individuals, connected in a social network, who are members of a committee, and suppose that each individual has his own opinion on the matter at hand. How can this group of people reach \emph{consensus}? This is a central question in economic theory, especially for processes in which people repeatedly average their own opinions. This line of work, see e.g. \cite{AO,DVZ03,GJ10,Jac08}, is based on a model defined by DeGroot \cite{DeG74}. In this model, each person $i$ holds an opinion equal to a real number $x_i$, which might for example represent a position on a political spectrum. There is an undirected weighted graph $G = (V, E, w)$ representing a social network, and node $i$ is influenced by the opinions of his neighbors in $G$ (the influence of neighbor $j$ is stronger the higher $w_{ij}$ is). In each time step, node $i$ updates his opinion to be a weighted average of his current opinion with the current opinions of his neighbors. A variation of this model of interest to our study is due to Friedkin and Johnsen \cite{FriJoh90}. In \cite{FriJoh90} it is additionally assumed that each node $i$ maintains a persistent \emph{internal belief} $b_i$, which remains constant even as node $i$ updates his overall opinion $x_i$ through averaging. (See Section \ref{sec:game} for the formal framework.)

However, as recently observed by Bindel et al. \cite{bkoFOCS11}, consensus is hard to reach, the case of political opinions being a prominent example. The authors of \cite{bkoFOCS11} justify the absence of consensus by interpreting repeated averaging as a decentralized dynamics for selfish players. Consensus is not reached as players will not compromise further when this diminishes their \emph{utility}. Therefore, these dynamics will converge to an equilibrium in which players might disagree; Bindel et al. study the cost of disagreement by bounding the price of anarchy in this setting.

In this paper, we continue the study of \cite{bkoFOCS11} and ask the question of how quickly equilibria are reached by decentralized dynamics in opinion games. We focus on the setting in which players have only a finite number of strategies available. This is motivated by the fact that in many cases although players have personal beliefs which may assume a continuum of values, they only have a limited number of strategies available. For example, in political elections, people have only a limited number of parties they can vote for and usually vote for the party which is \emph{closer} to their own opinions. Motivated by several electoral systems around the world, we concentrate in this study on the case in which players only have two strategies available. This setting already encodes a number of interesting technical challenges as outlined below.

\subsection{Our contribution}
For the finite version of the opinion games considered in \cite{bkoFOCS11}, we firstly note that this is a potential game \cite{R73,MS96} thus implying that these games admit pure Nash equilibria. The set of pure Nash equilibria is then characterized (cf. Lemma~\ref{lemma:Nash}). We also notice the interesting fact that while the games in \cite{bkoFOCS11} have a price of anarchy of 9/8, our games have unbounded price of anarchy, thus implying that for finite games  disagreeing has far more deep consequences on the social cost. We additionally prove that the socially optimal profile is always a Nash equilibrium when the weights of the social network are integers, thus implying that the Price of Stability is $1$ in this case, and in general a tight bound of $2$ for the Price of Stability when edges have fractional weights. 

We then study decentralized dynamics for finite opinion games. We first consider the best-response dynamics, by proving that it quickly converges to pure Nash equilibria in the case of unweighted social networks. For  general weights, we prove that the convergence rate is polynomial in the number of players but exponential in the representation of the weights. We also prove that for a specific opinion game, there exists an exponentially-long sequence of best responses thus implying that convergence may be exponential in general. The upper bounds are proved by ``reducing'' an opinion game to a version of it in which the internal beliefs can only take certain values. The reduced version is equivalent to the original one, as long as best-response dynamics is concerned. Note that the convergence rate for the version of the game considered in \cite{bkoFOCS11} is unknown.

In real life, however, there is some noise in the decision process of players. Arguably, people are not fully rational. On the other hand, even if they were, they might not exactly know what strategy represents the best response to a given strategy profile due to the incapacity to correctly determine their utility functions. To model this, we study \emph{logit dynamics} \cite{blumeGEB93} for opinion games. Logit dynamics features a \emph{rationality level} $\beta \geq 0$ (equivalently, a noise level $1/\beta$) and each player is assumed to play a strategy with a probability which is proportional to the corresponding utility to the player and $\beta$. So the higher $\beta$ is, the less noise there is and the more the dynamics is similar to best-response dynamics. Logit dynamics for potential games defines a Markov chain that has a nice structure. As in \cite{afppSAGT10, afpppSPAA11j} we exploit this structure to prove bounds on the convergence rate of logit dynamics to the so-called \emph{logit equilibrium}. The logit equilibrium corresponds to the stationary distribution of the Markov chain. Intuitively, a logit equilibrium is a probability distribution over strategy profiles of the game; the distribution is concentrated around pure Nash equilibrium profiles\footnote{Thus, the solution concept of logit dynamics is different from the one associated with best-response dynamics.}. It is observed in \cite{afppSAGT10} how this notion enjoys a number of desiderata one would like solution concepts to have.

We prove a host of results on the convergence rate of logit dynamics that give a pretty much complete picture as $\beta$ varies. We give an upper bound in terms of the cutwidth of the graph modeling the social network. The bound is exponential in $\beta$ and the cutwidth of the graph, thus yielding an exponential guarantee for some topology of the social network. We complement this result by proving a polynomial upper bound when $\beta$ takes a small value. We complete the preceding upper bound in terms of the cutwidth with lower bounds. Firstly, we prove that in order to get an (essentially) matching lower bound it is necessary to evaluate the size of a certain subset of strategy profiles. When $\beta$ is big enough relatively to this subset then we can prove that the upper bound is tight for any social network (specifically, we roughly need $\beta$ bigger than $n \log n$ over the cutwidth of the graph). For smaller values of $\beta$, we are unable to prove a lower bound which holds for every graph. However, we prove that the lower bound holds in this case at both ends of the spectrum of possible social networks. In details, we look at two cases of graphs encoding social networks: cliques, which model monolithic, highly interconnected societies, and complete bipartite graphs, which model more sparse ``antitransitive'' societies. For these graphs, we firstly evaluate the cutwidth and then relate the latter to the size of the aforementioned set of states. This allows to prove a lower bound exponential in $\beta$ and the cutwidth of the graph for (almost) any value of $\beta$. As far as we know, no previous result was known about the cutwidth of a complete bipartite graph; this might be of independent interest. The result on cliques is instead obtained by generalizing arguments in \cite{lpwAMS08}.

To prove the convergence rate of logit dynamics to logit equilibrium we adopt a variety of techniques developed to bound the mixing time of Markov chains. To prove the upper bounds we use some spectral properties of the transition matrix of the Markov chain defined by the logit dynamics, and coupling of Markov chains. To prove the lower bounds, we instead rely on the concept of bottleneck ratio and the relation between the latter and mixing time. (The interested reader might refer to \cite{lpwAMS08} for a discussion of these concepts. Below, we give a quick overview of these techniques and state some useful facts.)

\subsection{Related works}
In addition to the papers mentioned above, our paper is related to the works on logit dynamics.
This dynamics is introduced by Blume \cite{blumeGEB93} and it is mainly adopted in the analysis of graphical coordination games \cite{ellisonECO93, youngTR00, msFOCS09}, in which players are placed on vertices of a graph embedding social relations and each player wants to coordinate with neighbors: we highlight that an unique game is played on every edge, whereas, for opinion games, we need different games in order to encode beliefs. Asadpour and
Saberi \cite{asWINE09} adopt the logit dynamics for analyzing a class of congestion games. However, none of these works evaluates the time the logit dynamics takes in order to reach the stationary distribution: this line of research is conducted in \cite{afppSAGT10, afpppSPAA11j}.

A number of papers study the efficient computation of (approximate) pure Nash
equilibria for $2$-strategy games, such as
\emph{party affiliation games}~\cite{FPT04,BBM09} and \emph{cut games}~\cite{BCK10}.
The class of games we study here contrasts with those in that for the games considered here, Nash equilibria
can be found in polynomial time (Observation~\ref{obs:NEinP}), so that our
interest is in the extent to which equilibria can be found easily with simple
decentralized dynamic processes.
Similarly to these works, we focus on a class of $2$-strategy games and study
efficient computation of pure Nash equilibria; additionally we also study the
convergence rate to logit equilibria.

Another related work is \cite{DM11} by Dyer and Mohanaraj. They study graphical games, called \emph{pairwise-interaction games}, and prove among other results, quick convergence of best-response dynamics for these games. However, our games do not fall in their class. The difference is that, in their case, there is a unique game being played on the edges of the graph; as noted above, we instead need a different game to encode the internal beliefs of the players.

\section{The game}\label{sec:game}
Let $G = (V, E)$ be a connected undirected graph
with $|V| = n$ and for each edge $e=(i,j) \in E$ let $w_{ij} > 0$ be its weight.
We set $w_{ij} = 0$ if $(i,j)$ is not an edge of $G$. Every vertex of the graph represents a player.
Each player $i$ has an \emph{internal belief} $b_i \in [0,1]$ and only two strategies or \emph{opinions} are available, namely $0$ and $1$. Motivated by the model in \cite{bkoFOCS11}, we define the utility of player $i$ in a strategy profile $\x \in \{0,1\}^n$ as
$$
u_i(\x) = - \left((x_i - b_i)^2 + \sum_{(i,j) \in E} w_{ij} (x_i - x_j)^2\right).
$$
We call such a game an $n$-player \emph{opinion game} on a graph $G$. Let $D_i(\x) = \sum_{j \colon x_i \neq x_j} w_{ij}$ be the sum of the weights of edges going from $i$ to players with the opposite opinion to $i$. Then
$$
 u_i(\x) = - (x_i - b_i)^2 - D_i(\x).
$$

\subsection{Potential function}
Let $D(\x) = \sum_{i,j~:~x_i \neq x_j} w_{ij}$ be the sum of the weights of \emph{discording edges} in the strategy profile $\x$, that is the weight of all edges in $G$ whose endpoints have different opinions.
Thus, $D(\x)=\frac{1}{2}\sum_i D_i(\x)$.

\begin{lemma}
 The function
 \begin{equation}
  \label{eq:pot}
  \Phi(\x) = \sum_i (x_i - b_i)^2 + D(\x)
 \end{equation}
 is an exact potential function for the opinion game described above.
\end{lemma}

\begin{proof}
Given a strategy profile $\x$, player $i$ experiences a non-positive utility or
equivalently, a non-negative cost defined as $c_i(\x) = - u_i(\x)$.
We show that in an opinion game if players minimize their cost, then the function $\Phi$ defined in \eqref{eq:pot} decreases. The difference in the cost to player $i$ when he switches from strategy $x_i$ to strategy $y_i$ is
 $$
  c_i(\x) - c_i(\x_{-i},y_i) = (x_i - b_i)^2 - (y_i - b_i)^2 + D_i(\x) - D_i(\x_{-i},y_i).
 $$
 The difference in the potential function between the two corresponding profiles is
 \begin{align*}
  \Phi(\x) - \Phi(\x_{-i},y_i) & = \sum_j (x_j - b_j)^2 + D(\x) - \sum_{j \neq i} (x_j - b_j)^2 - (y_i - b_i)^2 -D(\x_{-i},y_i)\\
  & = (x_i - b_i)^2 - (y_i - b_i)^2  + D(\x) - D(\x_{-i},y_i).
 \end{align*}
Discording edges not incident on $i$ are not affected by the deviation of player $i$.
That is, if we let $K_i(\x) = \sum_{j,k \neq i;~x_j \neq x_k} w_{jk}$
be the sum of the weights of these edges, then $K_i(\x) = K_i(\x_{-i},y_i)$.
The claim then follows since $D(\x) = K_i(\x) + D_i(\x)$ and $D(\x_{-i},y_i) = K_i(\x_{-i},y_i) + D_i(\x_{-i},y_i)$.
\end{proof}

A more convenient way to express the potential function above, useful in one of the proofs below, is the following: $\Phi(\x) = \sum_{e \in E} \Phi_e(\x)$, where, for an edge $e=(i,j)$,
\begin{equation}\label{eq:pot:edge}
 \Phi_e(\x) = \begin{cases}
               \alpha_e := \frac{b_i^2}{\Delta_i} + \frac{b_j^2}{\Delta_j}, & \text{if } x_i = x_j = 0;\\
               \beta_e := \frac{b_i^2}{\Delta_i} + \frac{(1 - b_j)^2}{\Delta_j} + w_{ij}, & \text{if } x_i = 0 \text{ and } x_j = 1;\\
               \gamma_e := \frac{(1 - b_i)^2}{\Delta_i} + \frac{b_j^2}{\Delta_j} + w_{ij}, & \text{if } x_i = 1 \text{ and } x_j = 0;\\
               \delta_e := \frac{(1 - b_i)^2}{\Delta_i} + \frac{(1 - b_j)^2}{\Delta_j}, & \text{if } x_i = x_j = 1.
              \end{cases}
\end{equation}
and $\Delta_i$ represents the degree of $i$.

Notice that the potential function $\Phi$ looks similar to (but is not the same as) the social cost
$$
 \SC(\x) = - \sum_{i = 1}^n u_i(\x) = \sum_i (x_i - b_i)^2 + 2D(\x).
$$

\begin{obs}[\cite{efgm13}]
\label{obs:min-cost}
The profile minimizing the social cost can be computed in polynomial time by a centralized algorithm.
\end{obs}
Specifically, Escoffier et al.~\cite{efgm13} showed that it corresponds to the $(s,t)$-cut of minimum weight in a suitably built graph.

\subsection{Nash equilibria, Price of Anarchy and Price of Stability}
We start with a simple observation about the centralized computation of Nash equilibria.
\begin{obs}\label{obs:NEinP}
Nash equilibria of the opinion games studied here, can be found
in polynomial time.
\end{obs}
Specifically, consider the following simple greedy algorithm. Start with the pure
profile {\bf x} where everyone plays 0. Check for players who prefer to play 1; notice that any
such player must play 1 in {\em any} Nash equilibrium; modify {\bf x} so that they play 1.
Repeat this until no player prefers to change his own strategy.

Notice that this algorithm finds an equilibrium
that maximizes the number of players
playing 0, and we could similarly find one that maximizes the number of players
playing 1. It does not necessarily find a socially optimal equilibrium, although
it follows from Theorem~\ref{thm:PoS1} below and Observation~\ref{obs:min-cost}
that when edge weights are integers, the lowest-cost equilibrium is computable in polynomial time.

We next give a characterization of Nash equilibria.
Let $B_i$ be the integer closer to the internal belief of the player $i$: that is, $B_i = 0$ if $b_i \leq 1/2$, $B_i = 1$ if $b_i > 1/2$. Moreover, let $W_i = \sum_{j} w_{ij}$ be the total weight of edges incident on $i$ and $W_i^s(\x) = \sum_{j \colon x_j = s} w_{ij}$ be the total weight of edges going from $i$ to players playing strategy $s$ in the profile $\x$.

The following lemma shows that, for every player, it is preferable to select the opinion closer to his own belief if and only if more than (almost) half of his (weighted) neighborhood has selected this opinion.
\begin{lemma}
\label{lemma:Nash}
In a Nash equilibrium profile $\x$, it holds that for each player $i$
$$
  x_i = \begin{cases}
         B_i,     & {\rm if}~~W_i^{B_i}(\x) \geq \frac{W_i}{2} - \delta;\\
         1 - B_i, & {\rm if}~~W_i^{B_i}(\x) \leq \frac{W_i}{2} - \delta;
        \end{cases}
$$
where $\delta = \frac{1}{2} - |B_i - b_i|$.
\end{lemma}
\begin{proof}
 Let us start by observing that $|B_i - b_i| = B_i + b_i - 2b_iB_i$ and then
 \begin{equation}
  \label{eq:B_inequality}
  (1 - B_i - b_i)^2 = (B_i - b_i)^2 - 2|B_i - b_i| + 1.
 \end{equation}
 Now, we first prove that a profile for which the above conditions hold is a Nash equilibrium, and then we prove that every other profile is not in equilibrium.
 
Let $\x$ be a profile for which the above conditions hold for every player and $i$ be one such player. We consider first the case that $W_i^{B_i}(\x) \geq W_i/2 - \delta$: then we have $x_i = B_i$.
There is no incentive for $i$ to play $1 - B_i$ since
\begin{align*}
 u_i(\x_{-i}, 1 - B_i) & = - \left[(1 - B_i - b_i)^2 + W_i^{B_i}(\x)\right]\\
 & \leq - \left[ (B_i - b_i)^2 + \left(\frac{W_i}{2} + \delta\right)\right]\\
 & \leq - \left[(B_i - b_i)^2 + \left(W_i - W_i^{B_i}(\x)\right)\right] = u_i(\x),
\end{align*}
where we used \eqref{eq:B_inequality} for the first inequality.
Similarly, we can prove that if $W_i^{B_i}(\x) \leq \frac{W_i}{2} - \delta$, and thus $x_i = 1 - B_i$, then $u_i(\x_{-i}, B_i) \leq u_i(\x)$. Hence, no player has incentive to switch his opinion in $\x$ ad thus $\x$ is a Nash equilibrium.

 Now consider a profile $\y$ for which the conditions above do not hold for some player $i$. It must be the case that $W_i^{B_i}(\y) \neq W_i/2 - \delta$. If $W_i^{B_i}(\y) > W_i/2 - \delta$, this means that $y_i = 1 - B_i$; similarly, if $W_i^{B_i}(\y) < \frac{W_i}{2} - \delta$, we have that $y_i = B_i$. However, it is immediate to check that in the former case $u_i(\y) < u_i(\y_{-i}, B_i)$ and in the latter case $u_i(\y) < u_i(\y_{-i}, 1 - B_i)$.
\end{proof}
Roughly speaking, Lemma~\ref{lemma:Nash} identifies the point at which a player's neighbors dictate his strategy and overcome his internal belief.

\paragraph{Price of Anarchy and Stability.}
\begin{obs}\label{obs:PoA}
The price of anarchy of the social network games studied here is unbounded.
\end{obs}
To see this, consider the opinion game on a clique where each player has internal belief $0$
and each edge has weight $1$: the profile where each player has opinion $0$ has social cost $0$.
By Lemma~\ref{lemma:Nash}, the profile where each player has opinion $1$ is a Nash
equilibrium and its social cost is $n > 0$. This is in sharp contrast with the
bound $9/8$ proved in \cite{bkoFOCS11}.

We complete this section by proving bounds on the Price of Stability. 
\begin{theorem}
For opinion games, the price of stability is $2$.
\end{theorem}
\begin{proof}
We begin by proving the upper bound of $2$ on the price of stability for any opinion game $\G$. Let $\x^\star$ be the profile of $\G$ minimizing the potential function and $OPT$ be the profile minimizing the social cost. Therefore, we have
\begin{align*}
\SC(\x^\star) = \Phi(\x^\star) + D(\x^\star) \leq 2 \cdot \Phi(\x^\star) \leq 2 \cdot \Phi(OPT) \leq 2 \cdot \SC(OPT).
\end{align*}
Since $\x^\star$ is a Nash equilibrium then the upper bound to the price of stability follows.

For the lower bound, consider an opinion game $\G$ defined on a star-shaped social network with $n+1$ nodes, with $n > 4$, where each edge is weighted $1/n$. Let each external node, but one, have belief $1$. The center and the remaining external node have instead belief $0$. 

We now argue that the social cost is minimized by the profile in which all nodes apart from the center play their own belief. Indeed, its cost is $1 + 2/n = (n+2)/n$ (we have one discording edge weighted $1/n$ and additionally the center has a cost of $1$ since he is playing the strategy opposite to his own belief). All profiles in which at least two nodes play the strategy opposite to their belief have cost at least $2 > 1+2/n$. All profiles in which there is only one player, different from the center, playing opposite to his belief will have $k >1$ discording edges and then a cost of $ 1 + 2k/n > (n+2)/n$. Finally, the profile in which each player plays his belief has $n-1$ discording edges for a social cost of $2(n-1)/n = (2n - 2)/n > (n+2)/n$.

Now we note that the latter profile, that we will call $\x$, is the unique Nash equilibrium of the game. Indeed, for all nodes but the center, it is a dominant strategy to play their belief (by doing so, they will have a cost of at most $1/n$ while by switching their cost would be at least $1$). But then by Lemma \ref{lemma:Nash}, the center has a strict incentive to play his belief as well since, letting $i$ be the center, $W_i^0(\x) = 1/n > 0=W_i/2 -\delta$.

Therefore, the price of stability of this game is $2(n-1)/(n+2)$, which approaches $2$ as $n$ increases.
\end{proof}

We show that above bound can be improved when edge weights are integer.
Indeed, it turns out that, in this special case, the profile that minimizes
the social cost is always a Nash equilibrium.

\begin{theorem}\label{thm:PoS1}
 For an opinion game with integer weights, the price of stability is $1$.
\end{theorem}
\begin{proof}
 Let $\G$ be an $n$-player opinion game and let $\x$ be the profile that minimizes the social cost of $\G$. Assume for a contradiction that $\x$ is not a Nash equilibrium. This means there is a player $i$ for which the condition of Lemma~\ref{lemma:Nash} does not hold, i.e., either $x_i = 1 - B_i$ and $W_i^{B_i}(\x) > W_i/2 - \delta$ or $x_i = B_i$ and $W_i^{B_i}(\x) < W_i/2 - \delta$. Let us consider the first case (the second one can be handled similarly): we will show that the profile $(\x_{-i}, B_i)$ achieves a social cost lower than $\x$ and thus a contradiction. We evaluate the difference between $c_j(\x_{-i}, B_i)$ and $c_j(\x)$ for each player $j$. If $j = i$, then
 $$
  c_i(\x_{-i}, B_i) - c_i(\x) = W_i - W_i^{B_i}(\x) + (B_i - b_i)^2 - W_i^{B_i}(\x) - (1 - B_i - b_i)^2 = W_i - 2W_i^{B_i}(\x) + 2|B_i - b_i| - 1,
 $$
 where we used \eqref{eq:B_inequality}. Consider now a neighbor $j$ of $i$ such that $x_j = B_i$. Then,
 $$
  c_j(\x_{-i}, B_i) - c_j(\x) = W_j - W_j^{B_i}(\x) - w_{ij} + (B_i - b_j)^2 - W_j + W_j^{B_i}(\x) - (B_i - b_j)^2 = - w_{ij}.
 $$
 For a neighbor $j$ of $i$ such that $x_j = 1 - B_i$, we obtain
 $$
  c_j(\x_{-i}, B_i) - c_j(\x) = W_j^{B_i}(\x) + w_{ij} + (1 - B_i - b_j)^2 - W_j^{B_i}(\x) - (1 - B_i - b_j)^2 = w_{ij}.
 $$
 Finally, note that players that are not in the neighborhood of $i$ have the same cost in both profiles. Thus, the difference between social costs is:
 \begin{align*}
  \SC(\x_{-i}, B_i) - \SC(\x) & = W_i - 2W_i^{B_i}(\x) + 2|B_i - b_i| - 1 - \sum_{\substack{j \colon (i,j) \in E\\x_j = B_i}} w_{ij} + \sum_{\substack{j \colon (i,j) \in E\\x_j = 1 - B_i}} w_{ij}\\
  & = 2 \left(W_i - 2W_i^{B_i}(\x) + |B_i - b_i|\right) - 1.
 \end{align*}
Observe that by definition of $B_i$, $|B_i - b_i| \leq 1/2$. We now distinguish two cases. 

If $|B_i - b_i| < 1/2$ then $W_i^{B_i}(\x) > W_i/2 - \delta$ implies $W_i^{B_i}(\x) \geq W_i/2$ for integer weights. Therefore,  $\SC(\x_{-i}, B_i) - \SC(\x)= 2 \left(W_i - 2W_i^{B_i}(\x) + |B_i - b_i|\right) - 1 \leq  2|B_i - b_i| - 1 < 0$ and this concludes the proof in this case. 

If $|B_i - b_i| = 1/2$ then $W_i^{B_i}(\x) > W_i/2 - \delta$ is equivalent to $W_i^{B_i}(\x) > W_i/2$ and similarly to the case above we can conclude $\SC(\x_{-i}, B_i) - \SC(\x) < 0$.
\end{proof}

\section{Best-response dynamics}
To prove the convergence rate of best-response dynamics, we use the following
definition, previously used in~\cite{DM11}.
\begin{definition}
 Two games $\G, \G'$ are \emph{pure best-response equivalent}\footnote{In~\cite{DM11} they just use ``equivalent''.}
 if they have the same sets of players and pure-strategies,
 and for any player and pure-strategy profile, that player's best response is the same in $\G$ as in $\G'$.
\end{definition}
We prove bounds on the time the best-response dynamics for a game $\G$ takes to converge by analyzing the dynamics on a game $\G'$ that is \pbre to $\G$ but such that beliefs are ``nicely'' distributed. We say that a belief $b \in [0,1]$ is \emph{threshold} for player $i$ in opinion game $\G$ if player $i$ with belief $b$ in $\G$ is indifferent between playing strategy $0$ and strategy $1$ for some strategies of the players other than $i$. 

\begin{lemma}\label{lemma:bre}
Given an opinion game $\G$ and a player $i$ in $\G$, we define a finite set ${\cal B}_i$
of numbers in $[0,1]$ as follows. Let ${\cal B}'_i$ contain 0 and 1 together with every
threshold belief of player $i$. Let ${\cal B}_i$ contains every element of ${\cal B}'_i$,
and in addition, for every pair of consecutive elements $b'_1$, $b'_2$ of ${\cal B}'_i$, let ${\cal B}_i$
contain at least one element in the interval $(b'_1,b'_2)$ (for example, $\frac{1}{2}(b'_1+b'_2)$).

Then, any opinion game $\G$ is \pbre to an opinion game $\G'$ in which the beliefs of every
player $i$ in $\G$ have been replaced by an element of ${\cal B}_i$.
\end{lemma}
\begin{proof}
Fix an opinion game $\G$ and player $i$. Let $\G'$ be an opinion game defined on the same social network as $\G$.
For each player $i$ in $\G$, if his belief $b_i$ is threshold, then it belongs to ${\cal B}_i$,
and we keep it the same in $\G'$.
If $b_i$ is not a threshold belief, then in $\G'$ it is replaced by one of the elements of ${\cal B}_i$ that lies in
the subinterval bounded by 2 consecutive elements of ${\cal B}'_i$ containing $b_i$.

We claim that $\G'$ constructed this way, is \pbre to $\G$.
Consider a player $i$, and note that for a pair of beliefs $b_i$ and $b'_i$ to result in
different best responses, there must be a strategy profile for the remaining players
for which the best response under belief $b_i$ is opposite to the one under belief $b'_i$.
Note however that a player's beliefs have been changed in a way that ensures that there is no
such pair of pure-strategy profiles.
\end{proof}

\subsection{A special case: unitary weights}
We start by considering the special case in which $w_{ij} = 1$ for each edge $(i,j)$.
This helps to develop the ideas that we use to prove a bound for more general weights.

For a player $i$, consider ${\cal B}'_i$ as defined in Lemma~\ref{lemma:bre}.
In this special case, it is easy to see that $\mathcal{B}'_i = \{0,1\}$
if the neighborhood of $i$ has odd size and $\mathcal{B}'_i = \{0,\frac{1}{2},1\}$, otherwise.
Thus, by Lemma \ref{lemma:bre}, in both cases, $\G$ is \pbre to an opinion game $\G'$ where each player $i$ has belief in
$\mathcal{B}_i = \left\{0, \frac{1}{4}, \frac{1}{2}, \frac{3}{4}, 1 \right\}$.
The following theorem shows that the best-response dynamics quickly converges to a Nash equilibrium in $\G'$ and hence in $\G$.
\begin{theorem}
\label{thm:brd_special}
The best-response dynamics for an $n$-player opinion game converges to a Nash
equilibrium after a number of steps that is polynomial in $n$.
\end{theorem}
\begin{proof}
Let $\G$ be an $n$-player opinion game and let $\G'$ be a \pbre game having beliefs in
the set $\mathcal{B}_i = \left\{0, \frac{1}{4}, \frac{1}{2}, \frac{3}{4}, 1 \right\}$.
We show that best-response dynamics converges quickly on $\G'$.

We begin by observing that for every profile $\x$, we have $0 \leq \Phi(\x) \leq W + n$, where $W = \sum_{(i,j) \in E} w_{ij} \leq n^2$. Thus, the theorem follows by showing that at each time step the cost of a player decreases by at least a constant value. Fix $\x_{-i}$, the opinions of players other than $i$, and let $x_i$ be the strategy currently played by player $i$ and $s$ be his best response. By definition of best response, we have $c_i(\x) > c_i(s, \x_{-i})$. We will show that
\begin{equation}\label{eq:BRD:convergence}
\begin{aligned}
 \frac{1}{2} & \leq \Phi(\x)-\Phi(s, \x_{-i}) = c_i(x_i,\x_{-i}) - c_i(s,\x_{-i}) = (x_i- b_i)^2-(s-b_i)^2+ D_i(\x) - D_i(s, \x_{-i})\\
 & = 2|x_i - b_i| - 1 + \Delta D,
\end{aligned}
\end{equation}
where the last equality follows by \eqref{eq:B_inequality} (with $x_i$ in place of $B_i$) and by using $\Delta D$ as a shorthand for $D_i(\x) - D_i(s, \x_{-i})$.
We distinguish three cases based on the value of $\Delta D$.

\smallskip \noindent\underline{If $\Delta D > 1$,} then it will be the case that $\Delta D \geq 2$. Since the difference between the squares is bounded from below by $-1$, then \eqref{eq:BRD:convergence} follows.

\smallskip \noindent\underline{If $-1 < \Delta D \leq 1$,} then it will be the case that $\Delta D \in\{0,1\}$. If $x_i = 0$, then $c_i(x_i,\x_{-i})>c_i(s,\x_{-i})$ implies $b_i > \frac{1 - \Delta D}{2}$ and, from $b_i \in \mathcal{B}_i$, $b_i \geq \frac{1 - \Delta D + 1/2}{2}$; for these values of $b_i$ and $x_i$ \eqref{eq:BRD:convergence} follows. Similarly, if $x_i = 1$, then $c_i(x_i,\x_{-i})>c_i(s,\x_{-i})$ implies $b_i < \frac{1 + \Delta D}{2}$ and, from $b_i \in \mathcal{B}_i$, $b_i \leq \frac{1 + \Delta D - 1/2}{2}$; for these values of $b_i$ and $x_i$ \eqref{eq:BRD:convergence} follows.

\smallskip \noindent\underline{If $\Delta D \leq -1$,} then we will reach a contradiction. Indeed, since the difference between the squares is bounded from above by $1$, we have $c_i(x_i,\x_{-i})\leq c_i(s,\x_{-i})$.
\end{proof}

\subsection{Finite-precision weights}
We now show how to extend the bound for the convergence of the best-response dynamics to opinion games whose edge weights have bounded precision $k$, i.e. they can be written with at most $k$ digit after the decimal point.

Given an opinion game $\G$,
we consider a game $\G'$ that is exactly the same as $\G$ except that
in $\G'$ the utility of player $i$ in the profile $\x$ is
$$
u'_i(\x) = 10^k \cdot u_i(\x).
$$
We say $\G'$ is the \emph{integer version} of $\G$.
Obviously, $\G$ is \pbre to $\G'$ and $\G'$ is a potential game with potential function $\Phi' = 10^k \cdot \Phi$.
Note that $\G'$ can be equivalently described as follows:
each player has two strategies, 0 and 1, and a personal belief $b_i$ as in $\G$;
for each edge $e$ of the social graph of $\G$, we set $w'_e = 10^k \cdot w_e$ and $D'_i(\x) = \sum_{j \colon x_i \neq x_j} w'_{ij}$.
Then
$$
 u'_i(\x) = - \left( 10^k (x_i - b_i)^2 + D'_i(\x)\right).
$$
Note that $w'_e$ is an integer for each edge $e$. Then we have the following lemma.
\begin{lemma}
 \label{lem:b_integer}
 Consider an opinion game $\G$ and for each player $i$ consider the set $\mathcal{B}'_i$ as defined in Lemma~\ref{lemma:bre}.
 Then for each $b \in \mathcal{B}'_i$, we have that $10^k \cdot 2b$ is an integer.
\end{lemma}
\begin{proof}
If $b \in \{0,1\}$, then the lemma trivially follows. As for $b$ being a threshold belief, we distinguish two cases: if $b \leq 1/2$, then, by Lemma \ref{lemma:Nash}, there is a profile $\x$ such that
$$
 b_i = \frac{1}{2} - \frac{1}{2} \sum_{j} w_{ij} + \sum_{j \colon x_j = 0} w_{ij}.
$$
Hence,
$$
 10^k \cdot 2b_i = 10^k - \sum_{j} w'_{ij} + 2 \sum_{j \colon x_j = 0} w'_{ij}.
$$
Since each term in the right-hand side of the last equation is an integer then so is also $10^k \cdot 2b_i$.
The case $b > 1/2$ can be handled similarly.
\end{proof}

Now we are ready for proving a bound on the convergence time of the best-response dynamics.
\begin{theorem}
\label{thm:brd_general}
The best-response dynamics for an $n$-player opinion game $\G$ whose edge weights have bounded precision $k$
converges to a Nash equilibrium in $\OO(10^k \cdot n^2 \cdot w_{\max})$, where $w_{\max}$ is the largest edge weight in $\G$.
\end{theorem}
\begin{proof}
Fix the opinion game $\G$ and for each player $i$ consider the set $\mathcal{B}'_i$ as defined in Lemma~\ref{lemma:bre}.
Consider, moreover, ${\cal B}_i$ containing every element of ${\cal B}'_i$,
and in addition, the element $\frac{1}{2}(b'_1+b'_2)$ for every pair of consecutive elements $b'_1$, $b'_2$ of ${\cal B}'_i$.

Let $\G'$ be an opinion game such that each player $i$ has belief in ${\cal B}_i$. From Lemma~\ref{lemma:bre}, $\G$ and $\G'$ are \pbre games.
Moreover consider $\G''$ the integer version of $\G'$, which is \pbre to $\G'$ and then to $\G$. Below all the notation defined so far uses a double prime when it refers to $\G''$. 
We show that best-response dynamics converges in $\OO(10^k \cdot n^2 \cdot w_{\max})$ steps on $\G''$.

We begin by observing that for every profile $\x$, we have $0 \leq \Phi''(\x) \leq 10^k \left(\sum_{(i,j) \in E} w_{ij} + n\right) = \OO(10^k \cdot n^2 \cdot w_{\max})$. Thus, the theorem follows by showing that at each time step the cost of a player decreases by at least a constant value. Fix $\x_{-i}$, the opinions of players other than $i$, and let $x_i$ be the strategy currently played by player $i$ and $s$ be his best response. By setting $c''_i(\x) = - u''_i(\x)$, from the definition of best response, it follows that $c''_i(\x) > c''_i(s, \x_{-i})$. We will show that
\begin{equation}\label{eq:BRD:convergence_general}
\begin{aligned}
 1/2 & \leq \Phi''(\x)-\Phi''(s, \x_{-i}) = c''_i(x_i,\x_{-i}) - c''_i(s,\x_{-i}) \\ & = 10^k \left((x_i- b_i)^2-(s-b_i)^2\right) + D''_i(\x) - D''_i(s, \x_{-i}) = 10^k \left( 2|x_i - b_i| - 1 \right) + \Delta D,
\end{aligned}
\end{equation}
where the last equality follows by \eqref{eq:B_inequality} (with $x_i$ in place of $B_i$) and by using $\Delta D$ as a shorthand for $D''_i(\x) - D''_i(s, \x_{-i})$.
We distinguish three cases based on the value of $\Delta D$.

\smallskip \noindent\underline{If $\Delta D > 10^k$,} then, since all edge weights are integers, it will be the case that $\Delta D \geq 10^k +1$. Since the difference between the squares is bounded from below by $-1$, then \eqref{eq:BRD:convergence_general} follows.

\smallskip \noindent\underline{If $-10^k < \Delta D \leq 10^k$,} then, since all edge weights are integers, we have $\Delta D \in\{-10^k + 1, \ldots, 10^k\}$. If $x_i = 0$, then $c''_i(x_i,\x_{-i})>c''_i(s,\x_{-i})$ implies $b_i > \frac{10^k - \Delta D}{2 \cdot 10^k}$.
Since $10^k \cdot 2b$ is an integer for each threshold belief $b$, the smallest threshold belief greater than $\frac{10^k - \Delta D}{2 \cdot 10^k}$ should be at least $\frac{10^k - \Delta D + 1}{2 \cdot 10^k}$. Hence, the first element of $\mathcal{B}_i$ greater than $\frac{10^k - \Delta D}{2 \cdot 10^k}$ will be at least
$$
\frac{1}{2} \left(\frac{10^k - \Delta D}{2 \cdot 10^k} + \frac{10^k - \Delta D + 1}{2 \cdot 10^k}\right) = \frac{10^k - \Delta D + 1/2}{2 \cdot 10^k}.
$$
Thus, if $x_i = 0$, then $b_i \geq \frac{10^k - \Delta D + 1/2}{2 \cdot 10^k}$ and \eqref{eq:BRD:convergence_general} follows.
Similarly, one can prove that if $\x_i = 1$, then $b_i \leq \frac{10^k + \Delta D - 1/2}{2 \cdot 10^k}$ and \eqref{eq:BRD:convergence_general} follows.

\smallskip \noindent\underline{If $\Delta D \leq -10^k$,} then we will reach a contradiction. Indeed, since the difference between the squares is bounded from above by $1$, we have $c_i(x_i,\x_{-i})\leq c_i(s,\x_{-i})$.
\end{proof}

The bound on the convergence time of the best response dynamics given in previous theorem can be very large if the edge weights are very large or very small (in which case, we need high precision). However, in the next section we will show that, in these cases, such a large bound cannot be avoided.

\subsection{Exponentially many best-response steps for general weights}
The following result builds
a game with an exponentially large gap between the largest and the smallest edge weight
for which
there exist exponentially long sequences of best responses, where the
choice of the player switching his strategy at each step is made by an adversary.
Thus it remains an open question whether exponentially-many steps may be required
if, for example, players were allowed to best-respond in round robin manner, or
if the best-responding player was chosen randomly at each step.
However this does establish that the potential function
on its own is insufficient to upper-bound the number of steps with any polynomial.
The construction uses graphs with bounded degree and pathwidth, with
all players having belief $\frac{1}{2}$.

\begin{theorem}
The best-response dynamics for opinion games may take exponentially-many steps.
\end{theorem}
\begin{proof}
In the following construction, all players have a belief of $\frac{1}{2}$.

We start by giving some preliminary definitions.
A {\em 6-gadget} $G$ is a set of 6 players $\{A, B, C, D, E, F\}$ with edges
$(A,B)$, $(B,C)$ and $(C,D)$ having weights $\epsilon$, $2\epsilon$,
$3\epsilon$ respectively
and edges $(D,E)$, $(B,F)$ and $(D,F)$, all weighting $4\epsilon$, for some $\epsilon > 0$.

Consider a 6-gadget and a new player $A_0$ with edges $(A_0,B)$ and $(A_0,D)$ of weight $4\varepsilon$. We say that $A_0$ is a \emph{switch} for the 6-gadget $G$, since it allows $G$ to switch between the opinion vectors $(0,0,0,0,0,1)$ and $(0,1,0,1,0,1)$. Specifically, if initially the players
$(A,B,C,D,E,F)$ have opinions $(0,0,0,0,0,1)$ and $A_0$ is set to $1$, then we can have the following best-response sequence in the 6-gadget, that will be named \emph{switch-on cycle}:
\[
(0,0,0,0,0,1) \rightarrow (0,1,0,0,0,1) \rightarrow (0,1,0,1,0,1).
\]
If instead the players
$(A,B,C,D,E,F)$ have opinions $(0,1,0,1,0,1)$ and $A_0$ is set to $0$, then we can have the following best-response sequence in the 6-gadget, that will be named \emph{switch-off cycle}:
\[
\begin{aligned}
 & (0,1,0,1,0,1) \rightarrow (1,1,0,1,0,1) \rightarrow (1,0,0,1,0,1) \rightarrow (0,0,0,1,0,1)
\rightarrow (0,0,1,1,0,1) \rightarrow (0,1,1,1,0,1)\\
& \rightarrow (1,1,1,1,0,1) \rightarrow (1,1,1,0,0,1) \rightarrow
(1,1,0,0,0,1) \rightarrow (1,0,0,0,0,1) \rightarrow (0,0,0,0,0,1).
\end{aligned}
\]
Notice that $A$'s opinion does not change in the switch-on cycle, whereas it follows the sequence $0\rightarrow 1\rightarrow 0\rightarrow 1\rightarrow 0$ during the switch-off cycle.

We now define the game.
Consider $n$ 6-gadgets $G_i$ with players $\{A_i, B_i, C_i, D_i, E_i, F_i\}$ and
edge weights parametrized by $\epsilon_i$, where $1\leq i\leq n$.
For each $i = 2, \ldots, n$, we connect $G_i$ with $G_{i-1}$ by having $A_{i-1}$ acting as a switch for $G_i$.
We finally add the switch player $A_0$ for $G_1$. Thus, the total number of players is $6n + 1$.

In order that the behavior of $A_i$ in gadget $G_i$, for $i = 1, \ldots, n-1$ is not influenced by the new edges, $(A_i, B_{i+1})$ and $(A_i, D_{i+1})$ of weight $4\epsilon_{i+1}$, we need that the weights of these edges are small enough with respect to the weight of the unique edge incident on $A_i$ in the gadget $G_i$, namely $(A_i,B_i)$ of weight $\epsilon_i$. Specifically, it is sufficient to set $\epsilon_i > 8 \epsilon_{i+1}$.
Hence, the largest edge-weight is $4\varepsilon_1$ and the smallest edge-weight is $\varepsilon_n$ and their ratio is greater than $4\cdot 8^{n-1} = 2^{3n-1}$.

Consider now the following starting profile:
players $B_1$ and $D_1$ have opinion 1, all players $F_i$ have opinion 1,
and all other players start with opinion 0.
Note that $G_1$ is in the starting configuration of a switch-off cycle.

We finally specify an exponentially-long sequence of best-response:
we start the switch-off cycle of $G_1$;
as long as $A_i$, for $i = 1, \ldots, n-1$, switches his opinion from $0$ to $1$,
we execute the switch-on cycle of $G_{i+1}$;
as long as $A_i$, for $i = 1, \ldots, n-1$, switches his opinion from $1$ to $0$,
we execute the switch-off cycle of $G_{i+1}$.
Note that the last two cases occur two times during the switch-off cycle of $G_i$.
Thus, $G_2$ goes through $2$ switch-on cycles and $2$ switch-off cycles,
$G_3$ goes through $4$ switch-on cycles and $4$ switch-off cycles and, hence,
$G_n$ goes through $2^{n-1}$ switch-on cycles and $2^{n-1}$ switch-off cycles.
\end{proof}
%
%
%

\section{Logit Dynamics for Opinion Games}
Let $\G$ be an opinion game as from the above; moreover, let $S=\{0,1\}^n$ denote the set of all strategy profiles. For two vectors $\x,\y \in S$, we denote with $H(\x,\y) = |\{i \colon x_i \neq y_i\}|$ the Hamming distance between $\x$ and $\y$. The \emph{Hamming graph} of the game $\G$ is defined as $\Ham = (S, \mathsf{E})$, where two profiles $\x = (x_1, \dots, x_n), \y = (y_1, \dots, y_n) \in S$ are adjacent in $\Ham$ if and only if $H(\x, \y)=1$.

The \emph{logit dynamics} for $\G$ runs as follows: at every time step (i) Select one player $i \in [n]$ uniformly at random; (ii) Update the strategy of player $i$ according to the \emph{Boltzmann distribution} with parameter $\beta$ over the set $S_i=\{0,1\}$ of his strategies. That is, a strategy $s_i \in S_i$ will be selected with probability
\begin{equation}\label{eq:updateprob}
\sigma_i(s_i \mid \x_{-i}) = \frac{1}{Z_i(\x_{-i})} \, e^{\beta u_i(\x_{-i}, s_i)},
\end{equation}
where $\x_{-i} \in \{0,1\}^{n-1}$ is the profile of strategies played at the current time step by players different from $i$, $Z_i(\x_{-i}) = \sum_{z_i \in S_i} e^{\beta u_i(\x_{-i}, z_i)}$ is the normalizing factor, and $\beta \geq 0$. As mentioned above, from \eqref{eq:updateprob}, it is easy to see that for $\beta = 0$ player $i$ selects his strategy uniformly at random, for $\beta > 0$ the probability is biased toward strategies promising higher payoffs, and for $\beta$ that goes to $\infty$ player $i$ chooses his best response strategy (if more than one best response is available, he chooses one of them uniformly at random).

The above dynamics defines a \emph{Markov chain} $\{ X_t \}_{t \in \mathbb{N}}$ with the set of  strategy profiles as state space, and where the probability $P(\x,\y)$ of a transition from profile $\x = (x_1, \ldots, x_n)$ to profile $\y = (y_1, \ldots, y_n)$ is zero if $H(\x,\y) \geq 2$ and it is $\frac{1}{n} \sigma_i(y_i \mid \x_{-i})$ if the two profiles differ exactly at player $i$. More formally, we can define the logit dynamics as follows.
\begin{definition}[Logit dynamics~\cite{blumeGEB93}]
Let $\G$ be an opinion game as from the above and let $\beta \geq 0$. The \emph{logit dynamics} for $\G$ is the Markov chain $\mathcal{M}_\beta = \left(\{ X_t \}_{t \in \mathbb{N}},S,P\right)$ where $S = \{0,1\}^n$ and
\begin{equation}\label{eq:transmatrix}
P(\x, \y) = \frac{1}{n} \cdot
\begin{cases}
\sigma_i(y_i \mid \x_{-i}), & \quad \mbox{ if } \y_{-i} = \x_{-i} \mbox{ and } y_i \neq x_i; \\
\sum_{i=1}^n \sigma_i(y_i \mid \x_{-i}), & \quad \mbox{ if } \y = \x; \\
0, & \quad \mbox{ otherwise;}
\end{cases}
\end{equation}
where $\sigma_i(y_i \mid \x_{-i})$ is defined in \eqref{eq:updateprob}.
\end{definition}

The Markov chain defined by \eqref{eq:transmatrix} is ergodic. Hence, from every initial profile $\x$ the distribution $P^t(\x, \cdot)$ of chain $X_t$ starting at $\x$ will eventually converge to a \emph{stationary distribution} $\pi$ as $t$ tends to infinity.\footnote{The notation $P^t(\x, \cdot)$, standard in Markov chains literature \cite{lpwAMS08}, denotes the probability distribution over states of $S$ after the chain has taken $t$ steps starting from $\x$.} As in \cite{afppSAGT10}, we call the stationary distribution $\pi$ of the Markov chain defined by the logit dynamics on a game $\G$, the \emph{logit equilibrium} of $\G$.
In general, a Markov chain with transition matrix $P$ and state space $S$ is said to be \emph{reversible} with respect to the distribution $\pi$ if, for all $\x,\y\in S$, it holds that
$
 \pi(\x) P(\x,\y)=\pi(\y) P(\y,\x).
$
If the chain is reversible with respect to $\pi$, then $\pi$ is its stationary distribution. Therefore when this happens, to simplify our exposition we simply say that the matrix $P$ is reversible. For the class of {potential games} the stationary distribution is the well known \emph{Gibbs measure}.
\begin{theorem}[\cite{blumeGEB93}]
\label{thm:gibbsPot}
 If $\G = ([n], \mathcal{S}, \mathcal{U})$ is a potential game with potential function $\Phi$, then the Markov chain given by \eqref{eq:transmatrix} is reversible with respect to the Gibbs measure
$
\pi(\x) = \frac{1}{Z} e^{-\beta \Phi(\x)}
$,
where
$Z = \sum_{\y \in S} e^{-\beta \Phi(\y)}$
is the normalizing constant.
\end{theorem}

\paragraph{Mixing time of Markov chains.}
One of the prominent measures of the rate of convergence of a Markov chain to its stationary distribution is the \emph{mixing time}. For a Markov chain with transition matrix $P$ and state space $S$, let us set
$$
d(t)=
\max_{\x\in S}
\tv{P^t(\x,\cdot) - \pi},
$$
where the {\em total variation distance} $\tv{\mu - \nu}$ between two probability distributions $\mu$ and $\nu$ on the same state space $S$ is defined as
$$
\tv{\mu - \nu}=\max_{A\subset S}
	|\mu(A)-\nu(A)| = \frac{1}{2} \sum_{\x \in S} | \mu(\x) - \nu(\x) | =
	\sum_{\substack{\x \in S \colon\\ \mu(\x) > \nu(\x)}} \left(\mu(\x) - \nu(\x)\right).
$$
For $0 < \varepsilon < 1/2$, the mixing time is defined as
$$
\tm(\varepsilon) =
\min \{t\in\mathbb{N} \colon d(t)\leq\varepsilon\}.
$$
It is usual to set $\varepsilon = 1/4$ or $\varepsilon = 1/2e$. If not explicitly specified, when we write $\tm$ we mean $\tm(1/4)$. Observe that $\tm(\varepsilon)\leq\lceil\log_2 \varepsilon^{-1}\rceil\tm$.

\bigskip Next we bound the mixing time of the logit dynamics for an opinion game. Note that, bounds on the mixing time of the logit dynamics for general potential games has been given in~\cite{afpppSPAA11j}. However, these bounds are not tight for the class of opinion game. Moreover, they do not highlight the existing connection between the convergence time of the logit dynamics and the topology of the social network underlying the opinion game.

\subsection{Techniques}
To derive our bounds, we employ several different techniques: \emph{Markov chain coupling} and \emph{spectral techniques} for the upper bound and \emph{bottleneck ratio} for the lower bound. They are well-established techniques for bounding the mixing time; we next summarize them.

\subsubsection{Markov chain coupling}\label{subsubsec:coupling}
A {\em coupling} of two probability distributions $\mu$ and $\nu$ on $S$ is a pair of random variables $(X,Y)$ defined on $S \times S$ such that the marginal distribution of $X$ is $\mu$ and the marginal distribution of $Y$ is $\nu$. A {\em coupling of a Markov chain} $\mathcal{M}$ on $S$ with transition matrix $P$ is a process $(X_t,Y_t)_{t=0}^\infty$ with the property that $X_t$ and $Y_t$ are both Markov chains with transition matrix $P$ and state space $S$. 
When the two coupled chains $(X_t,Y_t)_{t=0}^\infty$ start at $(X_0,Y_0) = (\x,\y)$, we write $\Prob{\x,\y}{\cdot}$ and $\Expec{\x,\y}{\cdot}$ for the probability and the expectation on the space $S \times S$. We denote by $\tc$ the first time the two chains meet; that is,
$$
\tc=\min\{t \colon X_t=Y_t\}.
$$
We also consider only couplings of Markov chains with the property that $X_s=Y_s$ for $s\geq\tc$.

Recall that $\Ham=(S,\mathsf{E})$ is the Hamming graph; 
for $\x,\y\in S$, we denote by $\rho(\x,\y)$ the 
length of the shortest path in $\Ham$ between $\x$ and $\y$.
The following theorem says that it is sufficient to define a coupling only for pairs of Markov chains starting from states \emph{adjacent} in $\Ham$ and an upper bound on the mixing time can be obtained if each of these couplings contracts their distance on average.
\begin{theorem}[Path Coupling~\cite{BubleyDyer97}]
\label{theorem:pathcoupling}
Suppose that for every edge $(\x,\y)\in \mathsf{E}$ a coupling $(X_t,Y_t)$ of $\mathcal{M}$ with $X_0=\x$ and $Y_0=\y$ exists such that $\Expec{\x,\y}{\rho(X_1,Y_1)} \leq e^{-\alpha}$ 
for some $\alpha > 0$.
Then
$$
\tm(\varepsilon) \leq \frac{\log(\diam(\Ham)) + \log(1/\varepsilon)}{\alpha}
$$
where $\diam(\Ham)$ is the 
diameter of $\Ham$.
\end{theorem}
We here describe only the coupling that we use in the proof of Theorem~\ref{th:small_beta} below. Our exposition follows the one in \cite{afppSAGT10}.
For every pair of strategy profiles $\x = (x_1, \dots, x_n), \y = (y_1, \dots, y_n) \in \{0,1\}^n$ we define a coupling $(X_1,Y_1)$ of two copies of the Markov chain for which $X_0 = \x$ and $Y_0 = \y$. The coupling proceeds as follows: first, pick a player $i$ uniformly at random; then, update the strategies $x_i$ and $y_i$ of player $i$ in the two chains, by setting
$$
(x_i,y_i)=\begin{cases}
	(0,0), & \text{with probability }
		\min\{\sigma_i(0 \mid \x),\sigma_i(0 \mid \y)\};\cr
	(1,1), & \text{with probability }
		\min\{\sigma_i(1 \mid \x),\sigma_i(1 \mid \y)\};\cr
	(0,1), & \text{with probability }
		\sigma_i(0 \mid \x)-\min\{\sigma_i(0 \mid \x),\sigma_i(0 \mid \y)\};\cr
	(1,0), & \text{with probability }
		\sigma_i(1 \mid \x)-\min\{\sigma_i(1 \mid \x),\sigma_i(1 \mid \y)\}.
	\end{cases}
$$
Observe that for every player $i$, at most one of the updates $(x_i,y_i)=(0,1)$ and $(x_i,y_i)=(1,0)$ has positive probability.
Moreover, if $\sigma_i(0 \mid \x)=\sigma_i(0 \mid \y)$ and player $i$ is chosen, then, after the update, we have $x_i=y_i$.

For the path coupling technique (see Theorem~\ref{theorem:pathcoupling}), the coupling described above is applied only to pairs of starting profiles which are adjacent in $\Ham$.

\subsubsection{Relaxation time and spectral techniques}
Another important measure related to the convergence of Markov chains is given by the \emph{relaxation time}. Let $P$ be the transition matrix of a Markov chain with finite state space $S$ and let us label the eigenvalues of $P$ in non-increasing order
$
\lambda_1\geq \lambda_2 \geq \dots \geq \lambda_{|S|}.
$
It is well-known (see, for example, Lemma~12.1 in~\cite{lpwAMS08}) that $\lambda_1 = 1$ and, if $P$ is irreducible and aperiodic, then $\lambda_2<1$ and $\lambda_{|S|}>-1$. We set $\lambda^\star$ as the largest eigenvalue in absolute value other than $\lambda_1$, that is,
$
\lambda^\star = \max_{i=2, \ldots, |S|} \left\{ |\lambda_i| \right\}.
$
The {\em relaxation time} $\tr$ of a Markov chain $\mathcal{M}$ is defined as
$$
\tr = {\frac{1}{1-\lambda^\star}}.
$$
The relaxation time is related to the mixing time by the following theorem (see, for example, Theorems 12.3 and 12.4 in \cite{lpwAMS08}).
\begin{theorem}[Relaxation time]\label{theorem:relaxation}
Let $P$ be the transition matrix of a reversible, irreducible, and aperiodic Markov chain with state space
$S$ and stationary distribution $\pi$. Then
$$
(\tr-1)\log 2
\leq \tm\leq
\log\left({\frac{4}{\pimin}}\right) \tr,
$$
where
$\pimin=\min_{\x \in S} \pi(\x)$.
\end{theorem}

Bounds on relaxation time can be obtained by using the following lemma (see Corollary~13.24 in \cite{lpwAMS08}).
\begin{lemma}
 \label{comp_lemma}
 Let $P$ be the transition matrix of an irreducible, aperiodic and reversible Markov chain with state space $S$ and stationary distribution $\pi$. Consider the graph $\Ham = (S, \mathsf{E})$, where $\mathsf{E}=\{(\x,\y) \colon P(\x,\y) > 0\}$, and to every pair of states $\x, \y \in S$ we assign a path $\Gamma_{\x,\y}$ from $\x$ to $\y$ in $G$. We define
 $$
  \rho = \max_{e = (\z,\w) \in \mathsf{E}} \frac{1}{Q(e)} \sum_{\substack{ \x,\y \colon\\e \in \Gamma_{\x,\y}}} \pi(\x)\pi(\y)|\Gamma_{\x,\y}|.
 $$
 Then $\frac{1}{1 - \lambda_2} \leq \rho$.
\end{lemma}
We here mention a corollary of Lemma~\ref{comp_lemma} that will be useful for our results.
\begin{cor}
 \label{cor:comp_flow}
 Let $\G$ be an $n$-player opinion game with profile space $S$ and let $P$ and $\pi$ be the transition matrix and the stationary distribution of the logit dynamics for $\G$, respectively. For every pair of profiles $\x, \y$ we assign a path $\Gamma_{\x,\y}$ on the Hamming graph $\Ham$. Then
 $$
  \tr \leq 2n \max_{\substack{\z,\w \colon\\H(\z, \w) = 1\\\pi(\z) \leq \pi(\w)}} \frac{1}{\pi(\z)} \sum_{\substack{\x,\y \colon\\(\z,\w) \in \Gamma_{\x,\y}}} \pi(\x)\pi(\y)|\Gamma_{\x,\y}|.
 $$
\end{cor}
\begin{proof}
 In \cite{afpppSPAA11j} it is proved that all the eigenvalues of the transition matrix of the logit dynamics for potential games are non-negative. It follows then that $\tr = \frac{1}{1 - \lambda_2}$. Moreover, by reversibility of $P$, we have that for $\z,\w \in S$ such that $H(\z,\w)=1$ and $\pi(\z) \leq \pi(\w)$ it holds:
 $$
  Q(\z,\w) = \pi(\z) P(\z,\w) \geq \frac{\pi(\z)}{2n}.
 $$
 Thus, the claim follows from Lemma~\ref{comp_lemma}.
\end{proof}

\subsubsection{Bottleneck ratio}
Finally, an important concept to establish our lower bounds is represented by the \emph{bottleneck ratio}. Consider an ergodic Markov chain with finite state space $S$, transition matrix $P$, and stationary distribution $\pi$. The probability distribution $Q(\x, \y) = \pi(\x)P(\x,\y)$  is of particular interest and is sometimes called the \emph{edge stationary distribution}. Note that if the chain is reversible then $Q(\x, \y)= Q(\y, \x)$. For any $L \subseteq S$, we let $Q(L, S \setminus L)=\sum_{\x \in L, \y \in S \setminus L} Q(\x,\y)$. The bottleneck ratio of $L \subseteq S$, $L$ non-empty, is
$$
B(L) = \frac{Q(L,S \setminus R)}{\pi(L)}.
$$
We use the following theorem to derive lower bounds to the mixing time (see, for example, Theorem~7.3 in \cite{lpwAMS08}).
\begin{theorem}[Bottleneck ratio]\label{theorem:bottleneck}
Let $\mathcal{M} = \{ X_t \colon t \in \mathbb{N} \}$ be an irreducible and aperiodic Markov
chain with finite state space $S$,
transition matrix $P$, and stationary distribution $\pi$.
Then the mixing time is
$$
\tm \geq \max_{L \colon \pi(L) \leq 1/2} \frac{1}{4 B(L)}.
$$
\end{theorem}

\subsection{Upper bounds}
\subsubsection{\texorpdfstring{For every $\beta$}{For every beta}}
Consider the bijective function $\sigma \colon V \rightarrow \{1, \ldots, |V|\}$ representing an ordering of vertices of $G$. Let $\mathcal{L}$ be the set of all orderings of vertices of $G$ and set $V^\sigma_i = \{v \in V \colon \sigma(v) < i\}$. Moreover, for any partition $(L, R)$ of $V$ let $W(L,R)$ be the sum of the weights of edges that have an endpoint in $L$ and the other one in $R$.
Then, the (weighted) \emph{cutwidth} of $G$ is
\begin{equation*}
 \cw(G) = \min_{\sigma \in \mathcal{L}} \max_{1 < i \leq |V|} W(V^\sigma_i, V \setminus V^\sigma_i).
\end{equation*}
\begin{theorem}
 \label{th:all_graph}
 Let $\G$ be an $n$-player opinion game on a graph $G=(V,E)$. The mixing time of the logit dynamics for $\G$ is
 $$
  \tm \leq (1 + \beta) \cdot \poly{n,w_{\max}} \cdot e^{\beta\Theta(\cw(G))}.
 $$
\end{theorem}
\begin{proof}
This proof is a generalization of a similar proof given by Berger et al.~\cite{bkmp2005} and by Auletta et al.~\cite{afpppSPAA11j}.

Consider the ordering of vertices of $G$ that obtains the cutwidth. Fix $\x, \y \in S$ and let $v_1, v_2, \dots, v_d$ denote the indices (according to this ordering) of the vertices at which the profiles $\x$ and $\y$ differ; we consider the path $\Gamma_{\x,\y} = (\x^0, \x^1, \ldots, \x^d)$ on $\Ham$, where
$$
 \x^i = \left(y_1, \dots, y_{v_{i+1} - 1}, x_{v_{i + 1}}, \dots, x_n\right).
$$
(Above, we assume $v_{d + 1} = n + 1$). Notice that $\x^0 = \x$, $\x^d = \y$ and $|\Gamma_{\x,\y}| \leq n$. For every edge $\xi = (\x^i, \x^{i+1})$ of $\Ham$, we consider the function $\Lambda_\xi$ that assigns to every pair of profiles $\x, \y$ such that $\xi \in \Gamma_{\x,\y}$, the following new profile
$$
 \Lambda_\xi(\x, \y) = \begin{cases}
                       \left(x_1, \dots, x_{v_{i+1} - 1}, y_{v_{i + 1}}, y_{v_{i + 1} + 1}, \dots, y_n\right) & \text{if } \pi(\x^i) \leq \pi(\x^{i+1});\\
                       \left(x_1, \dots, x_{v_{i+1} - 1}, x_{v_{i+1}}, y_{v_{i + 1} + 1}, \dots, y_n\right) & \text{otherwise,}
                      \end{cases}
$$
where $\pi$ denotes the stationary distribution (cf. Theorem~\ref{thm:gibbsPot}). It is easy to see that $\Lambda_\xi$ is an injective function: indeed, since $\xi$ is known, if $\pi(\x^i) \leq \pi(\x^{i+1})$, then we can retrieve $v_{i + 1}$, that is the first vertex where $\x^i$ and $ \x^{i+1}$ differ and thus, selecting the first $v_{i + 1} - 1$  vertices from $\Lambda_\xi(\x, \y)$ and the remaining ones from $\x^i$ we are able to reconstruct $\x$ and, similarly, selecting the first $v_{i + 1} - 1$  vertices from $\x^i$ and the remaining ones from $\Lambda_\xi(\x, \y)$ we are able to reconstruct $\y$. Similarly, if $\pi(\x^i) > \pi(\x^{i+1})$, we can retrieve $v_{i + 1}$ and we can reconstruct $\x$ and $\y$ from $\Lambda_\xi(\x, \y)$ and $\x^{i+1}$.

Let $E^\star = \{ (j, k) \in E \colon j < v_{i + 1} \text{ and } k \geq v_{i + 1}\}$: observe that $\sum_{(j,k) \in E^\star} w_{jk} \leq \cw(G)$. For any edge $e = (j,k) \in E^\star$, for every $\x, \y \in S$ and for every $\xi = (\x^i, \x^{i+1}) \in \Gamma_{\x,\y}$, we distinguish two cases:

\noindent \underline{If $x_j = y_j$ or $x_k = y_k$,} for all available values of $x_j, y_j, x_k$ and $y_k$ we show
$$
 \Phi_{e}(\x) + \Phi_{e}(\y) - \Phi_{e}(\bot_{\x^i, \x^{i+1}}) - \Phi_{e}(\Lambda_\xi(\x, \y)) = 0,
$$
where $\bot_{\x^i,\x^{i+1}} = \arg\min\{\pi(\x^i),\pi(\x^{i+1})\}$. Firstly, assume that $x_j=y_j$ and $\bot_{\x^i,\x^{i+1}}= \x^i$ which in turns implies that $\Lambda_\xi(\x, \y)=(x_1, \dots, x_{v_{i+1} - 1}, y_{v_{i + 1}}, y_{v_{i + 1} + 1}, \dots, y_n)$. We have:
\begin{align*}
  \Phi_{e}(\x) + \Phi_{e}(\y) - \Phi_{e}(\bot_{\x^i, \x^{i+1}}) - \Phi_{e}(\Lambda_\xi(\x, \y)) &= \Phi_{e}(x_j, x_k) + \Phi_{e}(y_j, y_k) - \Phi_{e}(y_j, x_k) - \Phi_{e}(x_j, y_k) \\ &= \Phi_{e}(x_j, x_k) + \Phi_{e}(x_j, y_k) - \Phi_{e}(x_j, x_k) - \Phi_{e}(x_j, y_k)=0.
\end{align*}
It is not hard to check that the same is true for all the other possible cases arising.

\noindent \underline{If $x_j \neq y_j$ and $x_k \neq y_k$,} similarly to the above, it is not hard to see that for all available values of $x_j, y_j, x_k$ and $y_k$
$$
 \Phi_{e}(\x) + \Phi_{e}(\y) - \Phi_{e}(\bot_{\x^i, \x^{i+1}}) - \Phi_{e}(\Lambda_\xi(\x, \y)) = \pm (\alpha_e + \delta_e - \beta_e - \gamma_e) = \pm 2w_e,
$$
where $\alpha_e, \beta_e, \gamma_e$ and $\delta_e$ are defined in \eqref{eq:pot:edge}. Moreover for $e = (j,k) \in E \setminus E^\star$ it holds:
\[
\Phi_{e}(\x) + \Phi_{e}(\y) - \Phi_{e}(\bot_{\x^i, \x^{i+1}}) - \Phi_{e}(\Lambda_\xi(\x, \y)) = 0
\]
since, by construction, one of $\bot_{\x^i, \x^{i+1}}$ and $\Lambda_\xi(\x, \y)$ has $j$-th and $k$-th entry of $\x$ and the other has $j$-th and $k$-th entry of $\y$. Thus, we have that for every $\x, \y \in S$ and for every $\xi = (\x^i, \x^{i+1}) \in \Gamma_{\x,\y}$,
\begin{equation}
 \label{pot}
  \begin{split}
  \Phi(\x) + \Phi(\y) - \Phi(\bot_{\x^i, \x^{i+1}}) - \Phi(\Lambda_\xi(\x, \y)) & = \sum_{e \in E} \left(\Phi_{e}(\x) + \Phi_{e}(\y) - \Phi_{e}(\bot_{\x^i, \x^{i+1}}) - \Phi_{e}(\Lambda_\xi(\x, \y))\right)\\ &= \sum_{e \in E^\star} \left(\Phi_{e}(\x) + \Phi_{e}(\y) - \Phi_{e}(\bot_{\x^i, \x^{i+1}}) - \Phi_{e}(\Lambda_\xi(\x, \y))\right) \\
  & \geq -2\cw(G).
  \end{split}
\end{equation}
Now let $\xi^\star=(\z,\w)$ with $\pi(\z) \leq \pi(\w)$ be the edge of $\Ham$ for which
$
\sum_{\substack{\x,\y \colon\\ \xi^\star \in \Gamma_{\x,\y}}} \frac{\pi(\x)\pi(\y)}{\pi(\z)} |\Gamma_{\x,\y}|
$
is maximized. Applying Corollary~\ref{cor:comp_flow}, we obtain
\begin{align*}
 \tr & \leq 2n \sum_{\substack{\x,\y \colon\\ \xi^\star \in \Gamma_{\x,\y}}} \frac{\pi(\x)\pi(\y)}{\pi(\z)} |\Gamma_{\x,\y}| \leq 2n^2 \sum_{\substack{\x,\y \colon\\ \xi^\star \in \Gamma_{\x,\y}}} \frac{\pi(\x)\pi(\y)}{\pi(\z)\pi(\Lambda_{\xi^\star}(\x,\y))} \pi(\Lambda_{\xi^\star}(\x,\y)) \\
 & \leq 2n^2 e^{2\beta\cw(G)} \sum_{\substack{\x,\y \colon\\ \xi^\star \in \Gamma_{\x,\y}}} \pi(\Lambda_{\xi^\star}(\x, \y)) \leq 2n^2 e^{2\beta\cw(G)} \sum_{\x} \pi(\x) \leq  2n^2 e^{2\beta\cw(G)},
\end{align*}
where the third inequality follows from Theorem~\ref{thm:gibbsPot} and~\eqref{pot}, and the penultimate from the fact that $\Lambda_\xi$ is injective.

The theorem follows from Theorem~\ref{theorem:relaxation} and by observing that, since $\Phi(\x) \geq 0$ for any strategy profile $\x$, Theorem~\ref{thm:gibbsPot} implies
\begin{equation}
 \label{eq:pi_min}
 \begin{split}
   \log \left((\pimin/4)^{-1}\right) & = \log \left(4 \sum_\x e^{-\beta(\Phi(\x)-\Phi_{\max})}\right)\\
 & \leq \log \left(2^{n+2} \cdot e^{\beta\Phi_{\max}}\right) \leq \log \left(e^{n+2+\beta\Phi_{\max}}\right) = n + 2 + \beta\Phi_{\max},
 \end{split}
\end{equation}
where
$
 \Phi_{\max} = \max_\x \Phi(\x) \leq n + W.
$
\end{proof}

\subsubsection{\texorpdfstring{For small $\beta$}{For small beta}} The following theorem shows that for small values of $\beta$ the mixing time is polynomial. We remark that there are network topologies for which this theorem gives a bound higher than that guaranteed by Theorem~\ref{th:all_graph} on the values of $\beta$ for which the mixing time is polynomial.
\begin{theorem}
 \label{th:small_beta}
 Let $\G$ be an $n$-player opinion game on a connected graph $G$, with $n > 2$. Let $\Delta_{\max}$ be the maximum degree in the graph. If $\beta \leq 1/(w_{\max} \Delta_{\max})$, then the mixing time of the logit dynamics for $\G$ is $\OO(n \log n)$.
\end{theorem}
\begin{proof}
 Consider two profiles $\x$ and $\y$ that differ only in the strategy played by player $j$. W.l.o.g., we assume $x_j = 1$ and $y_j = 0$. We consider the coupling described in Section~\ref{subsubsec:coupling} for two chains $X$ and $Y$ starting respectively from $X_0=\x$ and $Y_0=\y$. We next compute the expected distance between $X_1$ and $Y_1$ after one step of the coupling.

 Let $N_i$ be the set of neighbors of $i$ in the opinion game. Notice that for any player $i$, $\sigma_i(0 \mid \x)$ only depends on $x_k$, for any $k \in N_i$, and $\sigma_i(0 \mid \y)$ only on $y_k$, for any $k \in N_i$. Therefore, since $\x$ and $\y$ only differ at position $j$, $\sigma_i(0 \mid \x)=\sigma_i(0 \mid \y)$ for $i \notin N_j$.

We start by observing that if position $j$ is chosen for update (this happens with probability $1/n$) then, by the observation above, both chains perform the same update. Since $\x$ and $\y$ differ only for player $j$, we have that the two chains are coupled (and thus at distance $0$). Similarly, if player $i \neq j$ with $i \notin N_j$ is selected for update (which happens with probability $(n - \Delta_j - 1)/n$) we have that both chains perform the same update and thus remain at distance $1$. Finally, let us consider the case in which $i \in N_j$ is selected for update. In this case, since $x_j=1$ and $y_j=0$, we have that $\sigma_i(0 \mid \x)\leq \sigma_i(0 \mid \y)$. Therefore, with probability $\sigma_i(0 \mid \x)$ both chains update position $i$ to $0$ and thus remain at distance $1$; with probability $\sigma_i(1 \mid \y) = 1 - \sigma_i(0 \mid \y)$ both chains update position $i$ to $1$ and thus remain at distance $1$; and with probability $\sigma_i(0 \mid \y)-\sigma_i(0 \mid \x)$ chain $X$ updates position $i$ to $1$ and chain $Y$ updates position $i$ to $0$ and thus the two chains go to distance $2$. By summing up, we have that the expected distance $E[\rho(X_1,Y_1)]$ after one step of coupling of the two chains is
 \begin{align*}
 E[\rho(X_1,Y_1)] & = \frac{n - \Delta_j - 1}{n} + \frac{1}{n}\sum_{i \in N_j} \left[ \sigma_i(0\mid\x)+1-\sigma_i(0\mid\y) + 2\cdot(\sigma_i(0\mid\y)-\sigma_i(0\mid\x))\right] \\
 & =
 \frac{n - \Delta_j - 1}{n}+
 \frac{1}{n}\cdot \sum_{i \in N_j}
 (1+\sigma_i(0\mid\y)-\sigma_i(0\mid\x))\\
 & =
 \frac{n-1}{n}+
 \frac{1}{n}\cdot \sum_{i \in N_j} (\sigma_i(0\mid\y)-\sigma_i(0\mid\x)).
 \end{align*}
 Let us now evaluate the difference $\sigma_i(0\mid\y)-\sigma_i(0\mid\x)$ for some $i \in N_j$. Recall that $W_i^s(\x)$ denotes the sum of the weights of edges connecting $i$ with neighbors that have opinion $s$ in the profile $\x$. Note that $W_i^0(\y) = W_i^0(\x) + w_{ij}$ and $W_i^1(\x) = W_i^1(\y) + w_{ij} = W_i - W_i^0(\x)$. For sake of compactness we will denote with $\ell$ the quantity $e^{\beta(2b_i - 1 + 2W_i^1(\x) - W_i)}$. By \eqref{eq:updateprob} we have
 $$
  \sigma_i(0 \mid \x) = \frac{e^{-\beta(b_i^2 + W_i^1(\x))}}{e^{-\beta(b_i^2 + W_i^1(\x))} + e^{-\beta((1 - b_i)^2 + W_i - W_i^1(\x))}} = \frac{1}{1 + \ell},
 $$
 and
 $$
  \sigma_i(0 \mid \y) = \frac{e^{-\beta(b_i^2 + W_i^1(\x) - w_{ij})}}{e^{-\beta(b_i^2 + W_i^1(\x) - w_{ij})} + e^{-\beta((1 - b_i)^2 + W_i - W_i^1(\x) + w_{ij})}} = \frac{1}{1 + \ell e^{-2w_{i,j}\beta}}.
 $$
 The function $\frac{1}{1 + \ell e^{-2w_{ij}\beta}} - \frac{1}{1 + \ell}$ is maximized for $\ell = e^{w_{ij}\beta}$. Thus
 $$
  \sigma_i(0 \mid \y) - \sigma_i(0 \mid \x) \leq \frac{1}{1 + e^{-\beta}} - \frac{1}{1 + e^{\beta}} = \frac{2}{1 + e^{-\beta}} - 1.
 $$
 By using the well-known approximation $e^{-w_{ij}\beta} \geq 1 - w_{ij}\beta$ and since by hypothesis $w_{ij}\beta \leq 1/\Delta_{\max}$, we have
 $$
  \sigma_i(0 \mid \y) - \sigma_i(0 \mid \x) \leq w_{ij}\beta \cdot \frac{1}{2 - w_{ij}\beta} \leq \frac{1}{\Delta_{\max}} \cdot \frac{\Delta_{\max}}{2\Delta_{\max} - 1}.
 $$
 We can conclude that the expected distance after one step of the chain is
$$
E[\rho(X_1,Y_1)] \leq \frac{n-1}{n} + \frac{1}{n} \cdot \frac{\Delta_j}{2\Delta_{\max} - 1} \leq \frac{n-1}{n} + \frac{2}{3n} = 1 - \frac{1}{3n} \leq e^{-\frac{1}{3n}}.
$$
where the second inequality relies on the fact that $\Delta_{\max} \geq 2$, since the social graph is connected and $n > 2$. Since $\diam(\Ham)=n$, by applying Theorem~\ref{theorem:pathcoupling} with $\alpha=\frac{1}{3n}$, we obtain the theorem.
\end{proof}

\subsection{Lower bound}
Recall that $\Ham$ is the Hamming graph on the set of profiles of an opinion games on a graph $G$. The following observation easily follows from the definition of cutwidth.
\begin{obs}
 For every path on $\Ham$ between the profile $\0 = (0, \ldots, 0)$ and the profile $\1 = (1, \ldots, 1)$ there exists a profile for which the weight of the discording edges is at least $\cw(G)$.
\end{obs}

From now on, let us write $\cw$ as a shorthand for $\cw(G)$, when the reference to the graph is clear from the context.
For sake of compactness, we set $\blf(\x) = \sum_i (x_i - b_i)^2$. We denote as $\blf^\star$ the minimum of $\blf(\x)$ over all profiles with $\cw$ discording edges.

Let $R_{\0}$ ($R_{\1}$) be the set of profiles $\x$ for which a path from $\0$ (resp., $\1$) to $\x$ exists in $\Ham$ such that every profile along the path has potential value less than $\blf^\star + \cw$. To establish the lower bound we use the technical result given by Theorem~\ref{theorem:bottleneck} which requires to compute the bottleneck ratio of a subset of profiles that is weighted at most a half by the stationary distribution. Accordingly, we set $R = R_{\0}$ if $\pi(R_{\0}) \leq 1/2$ and $R=R_{\1}$ if $\pi(R_{\1}) \leq 1/2$. (If both sets have stationary distribution less than one half, the best lower bound is achieved by setting $R$ to $R_{\0}$ if and only if $\Phi(\0) \leq \Phi(\1)$ since, in this case, $\blf(\0) \leq \blf(\1)$.) W.l.o.g., in the remaining of this section we assume $R = R_{\0}$.

\subsubsection{\texorpdfstring{For large $\beta$}{For large beta}}
Let $\partial R$ be the set of profiles in $R$ that have at least a neighbor $\y$ in the Hamming graph $\Ham$ such that $\y \notin R$. Moreover let $\mathcal{E}(\partial R)$ the set of edges $(\x,\y)$ in $\Ham$ such that $\x \in \partial R$ and $\y \notin R$: note that $|\mathcal{E}(\partial R)| \leq n |\partial R|$. The following lemma bounds the bottleneck ratio of $R$.

\begin{lemma}
\label{lemma:bottle_lb}
For the set of profiles $R$ defined above, we have
 $
  B(R) \leq n \cdot |\partial R| \cdot e^{-\beta(\cw + \blf^\star - \blf(\0))}.
 $
\end{lemma}
\begin{proof}
Since $\0 \in R$, it holds $\pi(R) \geq \pi(\0) = \frac{e^{-\beta\blf(\0)}}{Z}$. Moreover, by \eqref{eq:updateprob} we have
{\allowdisplaybreaks
\begin{align*}
  Q(R, \overline{R}) & = \sum_{\substack{(\x,\y) \in \mathcal{E}(\partial R): \\
  \y=(\x_{-i}, y_i)}} \frac{e^{-\beta\Phi(\x)}}{Z} \frac{e^{\beta u_i(\y)}}{e^{\beta u_i(\x)} + e^{\beta u_i(\y)}} \\
  & = \sum_{\substack{(\x,\y) \in \mathcal{E}(\partial R): \\ \y=(\x_{-i}, y_i)}} \frac{e^{-\beta\Phi(\x)}}{Z} \frac{e^{-\beta\Phi(\y)} e^{\beta(u_i(\x) + \Phi(\x))}}{e^{-\beta\Phi(\x)} e^{\beta(u_i(\x) + \Phi(\x))} + e^{-\beta\Phi(\y)} e^{\beta(u_i(\x) + \Phi(\x))}} \\ & = \frac{1}{Z} \sum_{(\x,\y) \in \mathcal{E}(\partial R)} \frac{e^{-\beta \Phi(\x)} e^{-\beta\Phi(\y)}}{e^{-\beta \Phi(\x)} + e^{-\beta \Phi(\y)}} = \frac{1}{Z} \sum_{(\x,\y) \in \mathcal{E}(\partial R)} \frac{e^{-\beta\Phi(\y)}}{1 + e^{\beta(\Phi(\x) - \Phi(\y))}} \\ & \leq  \frac{1}{Z} \sum_{(\x,\y) \in \mathcal{E}(\partial R)} e^{-\beta\Phi(\y)} \leq |\mathcal{E}(\partial R)| \cdot \frac{e^{-\beta(\blf^\star + \cw)}}{Z}.
 \end{align*}
 }%
 The second equality follows from the definition of potential function which implies $\Phi(\y)-\Phi(\x)=-u_i(\y)+u_i(\x)$ for $\x$ and $\y$ as above; last inequality holds because if by contradiction $\Phi(\y) < \blf^\star + \cw$ then, by definition of $R$, it would be $\y \in R$, a contradiction.
\end{proof}

From Lemma~\ref{lemma:bottle_lb} and Theorem~\ref{theorem:bottleneck} we obtain a lower bound to the mixing time of the opinion games that holds for every value of $\beta$, every social network $G$ and every vector $(b_1, \ldots, b_n)$ of internal beliefs. However, it is not clear how close this bound is to the one given in Theorem \ref{th:all_graph}. Nevertheless, by taking $b_i = 1/2$ for each player $i$ and $\beta$ high enough, we obtain the following theorem.
\begin{theorem}
\label{th:lb_high_beta}
 Let $\G$ be an $n$-player opinion game on a graph $G$. Then, there exists a vector of internal beliefs such that for $\beta = \Omega\left(\frac{n \log n}{\cw}\right)$ it holds
 $
  \tm \geq e^{\beta\Theta(\cw)}.
 $
\end{theorem}
\begin{proof}
If $b_i = 1/2$ for every player $i$, from Lemma~\ref{lemma:bottle_lb} and Theorem~\ref{theorem:bottleneck}, since $|\partial R| \leq 2^n$ then
 \[
  \tm \geq \frac{e^{\beta \cw}}{n 2^n} = e^{\beta \cw - n \log(2n)}=e^{\beta\Theta(\cw)}. \qedhere
 \]
\end{proof}

\subsubsection{\texorpdfstring{For smaller $\beta$}{For smaller beta}}
Theorem~\ref{th:lb_high_beta} gives an (essentially) tight lower bound for high values of $\beta$ for each network topology. It would be interesting to prove a matching bound also for lower values of the rationality parameter: in this section we prove such a bound for specific classes of graphs: complete bipartite graphs and cliques.

We start by considering the class of complete bipartite graphs $K_{m,m}$.
\begin{theorem}
\label{th:matching_lb}
 Let $\G$ be an $n$-player opinion game on $K_{m,m}$. Then, there exist a vector of internal beliefs and edge weights such that, for every $\beta = \Omega\left(\frac{1}{m}\right)$, we have
 $
  \tm \geq \frac{e^{\beta\Theta(\cw)}}{n}.
 $
\end{theorem}

To prove the theorem above, we start by evaluating the cutwidth of $K_{m,m}$: we focus on instances of the game with \emph{identical edge weights}. To simplify the exposition, we assume that $w_e=1$ for all the edges $e$ of $K_{m,m}$ and characterize the best ordering from which the cutwidth is obtained. We will denote with $A$ and $B$ the two sides of the bipartite graph.
\begin{claim}
 \label{claim:cutwidth_bipartite}
For identical edge weights, the ordering that obtains the cutwidth in $K_{m,m}$ is the one that selects alternatively a vertex from $A$ and a vertex from $B$. Moreover, the cutwidth of $K_{m,m}$ is $\lceil m^2/2 \rceil$.
\end{claim}
\begin{proof}
 Let $(T, V \setminus T)$ be a cut of the graph, we denote with $t$ the size of $T$: we also denote $t_A$ as the number of vertices of $A$ in $T$ and $t_B$ as the number of vertices of $B$ in $T$. Obviously, $t = t_A + t_B$. Given $t_A$ and $t_B$, the size of the cut $(T, V \setminus T)$ will be $t_A (m - t_B) + t_B (m - t_A)=mt - 2t_A(t-t_A)$. It is immediate to check that for every fixed $t$ the cut is minimized when $\lceil t/2 \rceil$ vertices of $T$ are taken from $A$ and the remaining ones from $B$. Therefore, the cutwidth is achieved by an ordering which selects alternatively vertices from the two sides of the graph and is then given by the maximum over $t$ of
 $$
  \bigg\lceil \frac{t}{2} \bigg\rceil \left(m - \bigg\lfloor \frac{t}{2} \bigg\rfloor \right) + \bigg\lfloor \frac{t}{2} \bigg\rfloor \left(m - \bigg\lceil \frac{t}{2} \bigg\rceil \right).
 $$
The above function is equal to $mt-\frac{t^2-1}{2}$ for $t$ odd and $m-t^2/2$ for $t$ even. Both these fuctions are maximized for $t = m$. However, this may be impossible to achieve when for example $t$ is odd and $m$ is even. Nevertheless, a simple case analysis on the parity of $m$ and $t$ shows that the maximum is achieved for $t=m-1,m,m+1$ when $m$ is even and for $t=m$ for $m$ odd, resulting in a cutwidth of $\lceil m^2/2 \rceil$.
\end{proof}

The following lemma gives a bound to the size of $\partial R$ for this graph.
\begin{lemma}
\label{lemma:bound_partialR}
 For the opinion game on the graph $K_{m,m}$ with $b_i = 1/2$ for every player $i$ and identical edge weights, there exists a constant $c_1$ such that $|\partial R| \leq e^{c_1\sqrt{\cw}}$.
\end{lemma}
\begin{proof}
Since $b_i = 1/2$ for every player $i$, we have that $\blf(\x) = n/4$ for every profile $\x$. Therefore, by definition of $R$, all profiles in $R$ (and therefore $\partial R$) have less then $\cw$ discording edges. Indeed, for $\x \in R$ we have $\blf(\x) + |D(\x)|=\Phi(\x) < \blf^\star + \cw$. Moreover, if a profile $\y$ has less then $\cw - m$ discording edges, then $\y$ is not in $\partial R$ as a state neighbor of $\y$ has at most $m-1$ additional discording edges.

Consequently, to bound the size of $\partial R$, we need to count the number of profiles in $R$ that have potential between $\blf^\star+\cw - m$ and $\blf^\star+\cw-1$ (i.e., the number of profiles with at least $\cw - m$ and at most $\cw-1$ discording edges). To count that, we consider two sets $L_0$ and $L_1$: we start by setting $L_0=V$ and $L_1=\emptyset$. We take vertices from $L_0$ and sequentially move them to $L_1$. We can think of $L_0$ as the set of vertices with opinion $0$ and $L_1$ as the set of vertices with opinion $1$: this way we can model a path from $\0$ to $\1$ in the Hamming graph. The number $M(t)$ of edges between $L_0$ and $L_1$ after $t$ moves is the number of discording edges in the social graph when vertices in $L_0$ have opinion $0$ and vertices in $L_1$ have opinion $1$. We have
 $$
  \bigg\lceil \frac{t}{2} \bigg\rceil \left(m - \bigg\lfloor \frac{t}{2} \bigg\rfloor \right) + \bigg\lfloor \frac{t}{2} \bigg\rfloor \left(m - \bigg\lceil \frac{t}{2} \bigg\rceil \right) \leq M(t) \leq mt,
 $$
where the lower bound follows from the structural proof of minimum cuts contained in Claim~\ref{claim:cutwidth_bipartite}.

 Let $t_1$ be the largest integer such that for all possible ways to choose $t_1 - 1$ vertices in $L_0$ and move them in $L_1$, the number of edges between $L_0$ and $L_1$ is less than $\cw - m$, i.e.
 $$
  (t_1 - 1)m < \cw - m \Rightarrow t_1 = \bigg\lfloor \frac{\cw}{m}\bigg\rfloor = \bigg\lfloor \frac{m}{2}\bigg\rfloor=\bigg\lfloor \frac n4\bigg\rfloor.
 $$
 Let $t_2$ be the smallest integer such that for all possible ways to move $t_2 + 1$ vertices from $L_0$ to $L_1$, the number of edges between $L_0$ and $L_1$ is at least $\cw$, i.e.
 $$
  \bigg\lceil \frac{t_2 + 1}{2} \bigg\rceil \left(m - \bigg\lfloor \frac{t_2 + 1}{2} \bigg\rfloor \right) + \bigg\lfloor \frac{t_2 + 1}{2} \bigg\rfloor \left(m - \bigg\lceil \frac{t_2 + 1}{2} \bigg\rceil \right) \geq \cw,
 $$
 that, as showed in Claim~\ref{claim:cutwidth_bipartite}, means $t_2 = m-2$ for $m$ even and $t_2=m-1$ for $m$ odd. Then, we can conclude $t_2 \leq m-1$.

 By the definition of $t_1$, all profiles with at most $t_1 - 1$ players with opinion $1$ are not in $\partial R$ and, by definition of $t_2$, all profiles with at least $t_2 + 1$ players with opinion $1$ are not in $R$. Thus, we have
 \begin{equation}
  \label{eq:bound_partialR}
  |\partial R| \leq \sum_{i = t_1}^{t_2} \binom{n}{i} \leq \sum_{i = t_1}^{t_2} \left(\frac{n \cdot e}{i}\right)^i \leq \sum_{i = t_1}^{t_2} \left(5e \right)^i = \frac{(5e)^{t_2+1}-(5e)^{t_1}}{5e-1} \leq (5e)^{t_2+1} \leq (5e)^m \leq\ e^{3m},
 \end{equation}
 where in the third inequality we used the fact that $i \geq t_1 > n/5$, in the penultimate the fact that $t_2 < m$ and lastly the fact that $5^m \leq e^{2m}$ for $m \geq 0$. The lemma follows since $m \leq \sqrt{2} \sqrt{\cw}$.
\end{proof}

\begin{proof}[Proof of Theorem~\ref{th:matching_lb}]
 If $b_i = 1/2$ for every player $i$, from Lemmata~\ref{lemma:bottle_lb} and~\ref{lemma:bound_partialR}, we have
 $$
  B(R) \leq n \cdot e^{c_1\sqrt{\cw}} \cdot e^{-\beta\cw} \leq n \cdot e^{-\beta\cw (1 - c_2)},
 $$
 where $c_2 = \frac{c_1 \sqrt{\cw}}{\beta \cw} < 1$ since by hypothesis $\beta > \frac{c_1}{\sqrt{\cw}} = \Omega(1/m)$; we also notice that $c_2$ goes to $0$ as $\beta$ increases. The theorem follows from Theorem~\ref{theorem:bottleneck}.
\end{proof}

We remark that it is possible to prove a result similar to Theorem~\ref{th:matching_lb} also for the clique $K_n$: the proof follows from a simple generalization of Theorem~15.3 in \cite{lpwAMS08} and by observing that the cutwidth of a clique is $\lfloor n^2/4 \rfloor$.

\section{Conclusions and open problems}
In this work we analyze two decentralized dynamics for binary opinion games: the best-response dynamics and the logit dynamics. As for the best-response dynamics we show that it takes time polynomial in the number of players to reach a Nash equilibrium, the latter being characterized by the existence of clusters in which players have a common opinion. On the other hand, for the logit dynamics we show polynomial convergence when the level of noise is high enough and that it increases as $\beta$ grows.

It is important to highlight, as noted above, that the convergence time of the two dynamics are computed with respect to two different equilibrium concepts, namely Nash equilibrium for the best-response dynamics and logit equilibrium for the logit dynamics. This explains why the convergence times of these two dynamics asymptotically diverge even though the logit dynamics becomes similar to the best response dynamics as $\beta$ goes to infinity.

Theorem \ref{th:all_graph} and \ref{th:lb_high_beta} which prove bounds to the convergence of logit dynamics can also be read in a positive fashion. Indeed, for social networks that have a bounded cutwidth, the convergence rate of the dynamics depends only on the value of $\beta$. (We highlight that checking if a graph has bounded cutwidth can be done in polynomial time \cite{tsbISAAC00}.) In general, we have the following picture: as long as $\beta$ is less than the maximum of (roughly) $\frac{\log n}{\cw}$ and $\frac{1}{\w_{\max}\Delta_{\max}}$ the convergence time to the logit equilibrium is polynomial. Moreover, Theorem~\ref{th:lb_high_beta} shows that for $\beta$ lower bounded by (roughly) $\frac{n \log n}{\cw}$ the convergence time to the logit equilibrium is super-polynomial. Then for some network topology, there is a gap in our knowledge which is naturally interesting to close.

In \cite{afppSODA12} the concept of metastable distributions has been introduced in order to predict the outcome of games for which the logit dynamics takes too much time to reach the stationary distribution for some value of $\beta$. It would be interesting to investigate existence and structure of such distributions for our opinion games.

\newpage

\bibliographystyle{plain}
\bibliography{meta4LD-upd}

\begin{thebibliography}{10}

\bibitem{AO}
Daron Acemoglu and Asuman Ozdaglar.
\newblock Opinion dynamics and learning in social networks.
\newblock {\em Dynamic Games and Applications}, 1(1):3--49, 2011.

\bibitem{asWINE09}
A.~Asadpour and A.~Saberi.
\newblock On the inefficiency ratio of stable equilibria in congestion games.
\newblock In {\em Proc. of the 5th International Workshop on Internet and
  Network Economics (WINE'09)}, volume 5929 of {\em Lecture Notes in Computer
  Science}, pages 545--552. Springer, 2009.

\bibitem{afppSODA12}
V.~Auletta, D.~Ferraioli, F.~Pasquale, and G.~Persiano.
\newblock Metastability of logit dynamics for coordination games.
\newblock In {\em Proceedings of the Twenty-Third Annual ACM-SIAM Symposium on
  Discrete Algorithms}, SODA '12, pages 1006--1024. SIAM, 2012.

\bibitem{afpppSPAA11j}
Vincenzo Auletta, Diodato Ferraioli, Francesco Pasquale, Paolo Penna, and
  Giuseppe Persiano.
\newblock Convergence to equilibrium of logit dynamics for strategic games.
\newblock {\em CoRR}, abs/1212.1884, 2012.
\newblock Preliminary version appeared in SPAA 2011.

\bibitem{afppSAGT10}
Vincenzo Auletta, Diodato Ferraioli, Francesco Pasquale, and Giuseppe Persiano.
\newblock Mixing time and stationary expected social welfare of logit dynamics.
\newblock {\em Theory of Computing Systems}, pages 1--38, 2013.

\bibitem{BBM09}
M.~Balcan, A.~Blum, and Y.~Mansour.
\newblock Improved equilibria via public service advertising.
\newblock In {\em SODA}, pages 728--737, 2009.

\bibitem{bkmp2005}
N.~Berger, C.~Kenyon, E.~Mossel, and Y.~Peres.
\newblock Glauber dynamics on trees and hyperbolic graphs.
\newblock {\em Probability Theory and Related Fields}, 131:311--340, 2005.

\bibitem{BCK10}
A.~Bhalgat, T.~Chakraborty, and S.~Khanna.
\newblock Approximating pure nash equilibrium in cut, party affiliation, and
  satisfiability games.
\newblock In {\em ACM Conference on Electronic Commerce}, pages 73--82, 2010.

\bibitem{bkoFOCS11}
D.~Bindel, J.~M. Kleinberg, and S.~Oren.
\newblock How bad is forming your own opinion?
\newblock In {\em FOCS}, pages 57--66, 2011.

\bibitem{blumeGEB93}
L.~E. Blume.
\newblock The statistical mechanics of strategic interaction.
\newblock {\em Games and Economic Behavior}, 5:387--424, 1993.

\bibitem{BubleyDyer97}
R.~Bubley and M.~E. Dyer.
\newblock Path {C}oupling: A technique for proving rapid mixing in {M}arkov
  chains.
\newblock In {\em Proceedings of the 38th Annual Symposium on Foundations of
  Computer Science (FOCS)}, pages 223--231. IEEE Computer Society, 1997.

\bibitem{DeG74}
M.~H. DeGroot.
\newblock Reaching a consensus.
\newblock {\em J. American Statistical Association}, 69, 1974.

\bibitem{DVZ03}
P.~M. DeMarzo, D.~Vayanos, and J.~Zweibel.
\newblock Persuasion bias, social influence, and unidimensional opinions.
\newblock {\em Quarterly Journal of Economics}, 2003.

\bibitem{DM11}
M.~Dyer and V.~Mohanaraj.
\newblock Pairwise-interaction games.
\newblock In {\em Proceedings of Automata, Languages and Programming - 38th
  International Colloquium, ICALP 2011, Part I}, volume 6755 of {\em Lecture
  Notes in Computer Science}, pages 159--170. Springer, 2011.

\bibitem{ellisonECO93}
G.~Ellison.
\newblock Learning, local interaction, and coordination.
\newblock {\em Econometrica}, 61(5):1047--1071, 1993.

\bibitem{efgm13}
B.~Escoffier, D.~Ferraioli, L.~Gourvès, and S.~Moretti.
\newblock Designing frugal best-response mechanisms for social network
  coordination games.
\newblock Under submission. Available at
  \url{http://www.lamsade.dauphine.fr/~ferraioli/papers/paris_paper.pdf}.

\bibitem{FPT04}
A.~Fabrikant, C.~H. Papadimitriou, and K.~Talwar.
\newblock The complexity of pure nash equilibria.
\newblock In {\em STOC}, pages 604--612, 2004.

\bibitem{SAGT12}
D.~Ferraioli, P.~W. Goldberg, and C.~Ventre.
\newblock Decentralized dynamics for finite opinion games.
\newblock In {\em SAGT}, pages 144--155, 2012.

\bibitem{FriJoh90}
N.~E. Friedkin and E.~C. Johnsen.
\newblock Social influence and opinions.
\newblock {\em Math. Sociology}, 15(3-4), 1990.

\bibitem{GJ10}
B.~Golub and M.~O. Jackson.
\newblock Naive learning in social networks: Convergence, influence and the
  wisdom of crowds.
\newblock {\em American Econ. J.: Microeconomics}, 2, 2010.

\bibitem{Jac08}
M.~O. Jackson.
\newblock {\em Social and Economic Networks}.
\newblock Princeton University Press, 2008.

\bibitem{lpwAMS08}
D.~Levin, Yuval P., and E.~L. Wilmer.
\newblock {\em Markov Chains and Mixing Times}.
\newblock American Mathematical Society, 2008.

\bibitem{MS96}
D.~Monderer and L.~S. Shapley.
\newblock Potential games.
\newblock {\em Games and Economic Behavior}, 14(1):124 -- 143, 1996.

\bibitem{msFOCS09}
A.~Montanari and A.~Saberi.
\newblock Convergence to equilibrium in local interaction games.
\newblock In {\em Proc. of the 50th Annual Symposium on Foundations of Computer
  Science (FOCS'09)}. IEEE, 2009.

\bibitem{youngTR00}
H.~Peyton~Young.
\newblock {\em The diffusion of innovations in social networks}, chapter in
  ``The Economy as a Complex Evolving System'', vol. III, Lawrence E. Blume and
  Steven N. Durlauf, eds.
\newblock Oxford University Press, 2003.

\bibitem{R73}
R.~W. Rosenthal.
\newblock A class of games possessing pure-strategy nash equilibria.
\newblock {\em International Journal of Game Theory}, 2(1):65--67, 1973.

\bibitem{tsbISAAC00}
D.~M. Thilikos, M.~J. Serna, and H.~L. Bodlaender.
\newblock Constructive linear time algorithms for small cutwidth and
  carving-width.
\newblock In {\em ISAAC}, pages 192--203, 2000.

\end{thebibliography}

\end{document}